\pdfoutput=1
\documentclass[12pt, letterpaper]{article}
\usepackage[left = 1in, right = 1in, top = 1in, bottom = 1.2in]{geometry}
\usepackage{amsfonts, amssymb, amsmath}
\usepackage{array}
\usepackage{booktabs}
\usepackage{natbib}
\usepackage{graphicx}
\usepackage{multirow}
\usepackage{amsmath,bm}
\usepackage{graphicx,psfrag,epsf}
\usepackage{enumerate}
\usepackage{natbib}
\usepackage{url} 

\usepackage{amssymb, amsfonts, amsthm}

\let\E\undefined
\DeclareMathOperator{\E}{E}
\let\pr\undefined
\DeclareMathOperator{\pr}{Pr}

\let\Pn\undefined

\DeclareMathOperator{\Pn}{\mathbb{P}_{\mathit{n}}}

\def\wt{\widetilde}
\def\wh{\widehat}
\def\wb{\overline}
\def\wub{\underline}

\def\L{\mathcal{L}}
\def\TB{\mathcal{T}_B}
\def\S{\mathcal{S}}
\def\Sc{\mathcal{S}^c}

\def\R{\mathbb{R}}
\def\Hgamma{\mathbb{H}_{\bm{\gamma}}}
\def\H{\mathbb{H}}

\let\var\undefined
\DeclareMathOperator{\var}{Var}

\let\T\undefined
\def\T{\mathrm{\scriptscriptstyle T}}
\DeclareMathOperator*{\argmax}{arg\,max}
\DeclareMathOperator*{\argmin}{arg\,min}

\newcommand{\RomanNum}[1]{\MakeUppercase{\romannumeral #1\relax}}
\newcommand{\romanNum}[1]{\MakeLowercase{\romannumeral #1\relax}}

\newcommand{\cc}[1]{\multicolumn{1}{c}{#1}}
\newcommand{\cellcenter}{\cc}

\def\symmdiff{\mathbin{\triangle}}

\theoremstyle{plain}
\newtheorem{theorem}{Theorem}
\newtheorem{proposition}{Proposition}
\newtheorem{lemma}[proposition]{Lemma}
\newtheorem{corollary}[proposition]{Corollary}

\theoremstyle{remark}
\newtheorem{remark}{Remark}
\newtheorem{assumption}{Assumption}

\title{Interpretable Dynamic Treatment Regimes}

\author{
  Yichi Zhang,
  Eric B. Laber,
  Anastasios Tsiatis,  
  and Marie Davidian\\ 
North Carolina State University}

\date{}

\begin{document}

\maketitle

\begin{abstract}
  Precision medicine is currently a topic of great interest in
  clinical and intervention science.  A key component of precision
  medicine is that it is evidence-based, i.e., data-driven, and
  consequently there has been tremendous interest in estimation of
  precision medicine strategies using observational or randomized
  study data.  One way to formalize precision medicine is through a
  treatment regime, which is a sequence of decision rules, one per
  stage of clinical intervention, that map up-to-date patient
  information to a recommended treatment.  An optimal treatment regime is
  defined as maximizing the mean of some cumulative clinical outcome
  if applied to a population of interest.  It is well-known that
  even under simple generative models an optimal treatment regime can
  be a highly nonlinear function of patient information.
  Consequently, a focal point of recent methodological research has
  been the development of flexible models for estimating optimal
  treatment regimes.  However, in many settings, estimation of an
  optimal treatment regime is an exploratory analysis intended to
  generate new hypotheses for subsequent research and not to directly
  dictate treatment to new patients.  In such settings, an estimated
  treatment regime that is interpretable in a domain context may be of
  greater value than an unintelligible treatment regime built using
  `black-box' estimation methods.  We propose an estimator of an
  optimal treatment regime composed of a sequence of decision rules,
  each expressible as a list of ``if-then'' statements that can be
  presented as either a paragraph or as a simple flowchart that is
  immediately interpretable to domain experts.  The discreteness of
  these lists precludes smooth, i.e., gradient-based, methods of
  estimation and leads to non-standard asymptotics.  Nevertheless, we
  provide a computationally efficient estimation algorithm, prove
  consistency of the proposed estimator, and derive rates of
  convergence.  We illustrate the proposed methods using a series of
  simulation examples and application to data from a sequential
  clinical trial on bipolar disorder.

\end{abstract}

\section{Introduction}
\label{sec:intro}
Precision medicine
is now almost universally recognized as a path to delivering the
best possible healthcare \citep[][]{collins2015new,
  ashley2015precision, jameson2015precision}.  Furthermore,
technological advancements and investment in big-data infrastructure
have made it possible to collect, store, and curate large amounts of
patient-level data to inform the practice of precision medicine
\citep[][]{krumholz2014big}.  Quantitative researchers have responded
with a surge of methodological developments aimed at `mathematizing'
precision medicine in the form of treatment regimes, a sequence of
decision rules, one per stage of clinical intervention,  that map
up-to-date patient information to a treatment recommendation; an
optimal treatment regime is defined as maximizing the mean of some
desirable clinical outcome if applied to a population of interest.  It
can be shown that even under the simplest generative models the
optimal regime is a nonlinear function of patient information
\citep[][]{robins2004optimal, schulte2014q-and-a,
  laber2014interactive}; consequently, to avoid model
misspecification, a recent trend is to apply flexible supervised
learning methods to estimate optimal treatment regimes.  These flexible methods
include direct-search using large-margin classifiers
\citep[][]{zhao2012estimating, zhao2014doubly, kang2014combining,
  zhao2015new, xu2015regularized}; $Q$-learning with non-parametric
regression models \citep[][]{qian2011performance,
  zhao2011reinforcement, moodie2013q-learning, kosorokKNN}; and
tree-based methods \citep[][]{zhang2012estimating, laber2015tree, zhang2015using,
  doove2015novel}. Further testament to the popularity of these
methods is that the Journal of the American
Statistical Association's Theory and Methods Invited Paper and
the Case Studies and Applications Invited Paper at the 2016 Joint
Statistical Meetings will feature non-parametric methods for estimating
treatment regimes \citep[][]{zhou2015residual, xu2015bayesian}.
 
Flexible estimation methods mitagate the risk of model
misspecification but potentially at the price of rendering the
estimated regime unintelligible.  This price is may be too high in
settings where the primary role of an estimated optimal regime is to
generate new scientific hypotheses or inform future research.  For
example, in the context of sequential multiple assignment clinical
trials \citep[SMARTs,][]{murphy2005experimental, lei2012smart}
estimation of an optimal treatment regime is typically included as a
secondary, exploratory analysis, as sizing the trial to ensure
high-quality estimation of an optimal regime is complex
\citep[][]{laber2016sizing}.  Tree-based regimes, like regression or
classification trees, offer flexibility while retaining
interpretability.  Here, we propose a method for estimation of an
optimal treatment regime that comprises a sequence of decision rules
each of which is represented as a sequence if-then statements mapping logical
clauses to treatment recommendations.  Decision rules of this form are
a special case of tree-based rules, known as decision lists
\citep[][]{rivest1987learning, marchand2005learning, letham2012building, wang2015falling, zhang2015using}, that
are immediately interpretable in a domain context as they can be
expressed in either flow-chart or paragraph form.  Thus, regimes of
this form are amenable to critique and examination by clinicians and
facilitate collaborative, iterative development of data-driven
precision medicine.  Furthermore, we shall show that despite the
structure imposed by the decision lists, they are sufficiently
expressive so as to provide high-quality regimes even under non-linear
generative models previously used in the literature to illustrate the
value of non-parametric estimation methods.

In addition to the clinical and scientific value of interpretable,
list-based regimes, the proposed work provides a number of important
methodological contributions.  Unlike existing tree-based methods for
estimating optimal treatment regimes, the proposed methodology applies
to problems with an arbitrary number of treatment stages and
treatments per stage.  In principle, robust policy-search
\citep[][]{zhang2013robust} could be used with CART
\citep[][]{breiman1984classification} to estimate a multi-stage,
tree-based treatment regime; however, this method relies on inverse
probability of treatment weighting which rapidly becomes unstable as
the number of treatment stages increases.  A second contribution is
that we prove that the proposed estimator is consistent for the
optimal regime within the class of list-based regimes and derive rates
of convergence for the proposed estimator.  These theoretical results
are non-trivial because the discreteness of the list precludes the use
of standard asymptotic approaches; to our knowledge these are first
results on convergence rates for decision lists and are therefore
of independent interest.  A third contribution is the
proposed estimation algorithm used to construct the decision lists at
each stage.  This algorithm reduces computation time of naive
recursive-splitting algorithm from $O(n^3)$ to $O(n\log\,n)$ where $n$
is the number of subjects in the sample.  Furthermore we modify the
splitting criteria proposed by \citet{zhang2015using} to avoid
(asymptotically) becoming stuck in a local mode.

In  Section~\ref{sec:method}, we describe list-based treatment regimes
and describe our estimation algorithm.  In Section~\ref{sec:theory}, we 
prove consistency of the proposed estimator and derive rates of
convergence.    In Section ~\ref{sec:simulation}, we demonstrate
the finite sample performance of the proposed method using 
simulation experiments.  We illustrate the proposed method
using data from a clinical trial in Section~\ref{sec:data-analysis} 
and make concluding remarks in Section~\ref{sec:discussion}.

\section{Methodology}
\label{sec:method}
\subsection{Framework}
Consider $n$ $i.i.d.$ observations collected
from a sequential clinical trial with $T$ stages;  the proposed
methodology also applies to observational data provided
that standard causal assumptions required for $Q$-learning
are satisfied 
\citep[see][for a statement of these
assumptions]{schulte2014q-and-a}.  
In the assumed setup 
the observed data are
$\{ (\bm{S}_{it}, A_{it}, Y_{it})\,:\, t = 1, \dots, T \}_{i=1}^{n}$,
which comprise $i.i.d.$ trajectories of the form 
$\{ (\bm{S}_{t}, A_{t}, Y_{t})\,:\, t = 1, \dots, T \}$
 where:
$\bm{S}_{t} \in \mathbb{R}^{p_t}$ is a vector of covariates
measured at the beginning of the $t$-th stage;
$A_{t} \in \mathcal{A}_t$ is the treatment
actually received during the $t$-th stage; and
$Y_{t} \in \mathbb{R}$ is a scalar outcome measured at the end of
the $t$-th stage.  Let $m_t = |\mathcal{A}_t|$ denote
the number of available treatment options at the $t$-th stage.
The final outcome of interest is the sum of
immediate outcomes, $Y = \sum_{t=1}^T Y_{t}$.  We assume that
larger values of $Y$ are better.  
Let $\bm{X}_{t}$ denote the information available to the decision
maker
at stage $t$ so that $\bm{X}_1 = \bm{S}_1$ and 
$\bm{X}_t  = (\bm{X}_{t-1}, A_{t}, \bm{S}_t)$ for $t > 1$.  
Let $\mathcal{X}_{t} \subset \mathbb{R}^{d_t}$ be the support of
$\bm{X}_{t}$, where $d_t = \sum_{s=1}^t p_s + 2(t-1)$ is the
dimension of $\bm{X}_{t}$. 

A treatment regime $\bm{\pi} = (\pi_1,\ldots, \pi_T)$ is a sequence
of functions $\pi_t:\mathcal{X}_t\rightarrow \mathcal{A}_t$ so that
under 
$\bm{\pi}$ a patient presenting with $\bm{X}_t = \bm{x}_t$ at stage
$t$  is recommended treatment $\pi_t(\bm{x}_t)$.  For simplicity,
we assume that all treatment are feasible for all patients; the
extension
to patient-specific sets of feasible treatments is straightforward 
\citep[][]{schulte2014q-and-a}.  For any
regime $\bm{\pi}$,  let $\mathbb{E}^{\bm{\pi}}$ denote expectation with 
respect to distribution induced by assigning treatments according
to $\bm{\pi}$.  Given a class of
regimes $\Pi$, an optimal regime satisfies, $\bm{\pi}^{\mathrm{opt}}
\in \Pi$ and
$\mathbb{E}^{\bm{\pi}^{\mathrm{opt}}}Y 
\ge \mathbb{E}^{\bm{\pi}}Y$ for all $\bm{\pi} \in\Pi$.   Our goal
is to construct an estimator of $\bm{\pi}^{\mathrm{opt}}$ when 
$\Pi$ is the class of list-based regimes.  
Each decision rule $\pi_t$ in a list-based regime has  the form:
\begin{align}
\label{eq:list}
&\texttt{If } \bm{x}_t \in R_{t1} \texttt{ then } a_{t1}; \nonumber \\
&\texttt{else if } \bm{x}_t \in R_{t2} \texttt{ then } a_{t2}; \nonumber \\
&\texttt{...} \nonumber \\
&\texttt{else if } \bm{x}_t \in R_{t L_t} \texttt{ then } a_{t L_t},
\end{align}
where: each $R_{t\ell}$ is a subset of $\mathcal{X}_t$ with the
restriction that
$R_{t L_t} = \mathcal{X}_t$;
$a_{t\ell} \in \mathcal{A}_t$;
$\ell=1,\ldots, L_t$; and $L_t$ is the length of $\pi_t$. 
Thus, a compact representation of $\pi_t$ is 
$\{(R_{t\ell}, a_{t\ell})\}_{\ell=1}^{L_t}$.  
To increase
interpretability, we restrict $R_{t\ell}$ to clauses involving
thresholding with at most two covariates, hence
$R_{t\ell}$ is an element of
\begin{align}
\label{eq:form}
\mathcal{R}_{t} = \{
& \mathcal{X}_t,\;
\{\bm{x} \in \mathcal{X}_t: x_{j_1} \leq \tau_1\},\;
\{\bm{x} \in \mathcal{X}_t: x_{j_1} > \tau_1\}, \nonumber \\
& \{\bm{x} \in \mathcal{X}_t: x_{j_1} \leq \tau_1 \text{ and } 
                              x_{j_2} \leq \tau_2\},\;
\{\bm{x} \in \mathcal{X}_t: x_{j_1} \leq \tau_1 \text{ and } 
                            x_{j_2} > \tau_2\}, \nonumber \\
& \{\bm{x} \in \mathcal{X}_t: x_{j_1} > \tau_1 \text{ and }   
                               x_{j_2} \leq \tau_2\},\;
\{\bm{x} \in \mathcal{X}_t: x_{j_1} > \tau_1 \text{ and }   
                            x_{j_2} > \tau_2\}: \nonumber \\ 
& 1 \leq j_1 < j_2 \leq d_t, 
  \tau_1, \tau_2 \in \mathbb{R} \},
\end{align}
where $j_1, j_2$ are indices and $\tau_1, \tau_2$ are thresholds.
We also impose an upper bound, $L_{\mathrm{max}}$,  on list length $L_t$
for all $t$.
Hence, the class of regimes of interest is 
$\Pi = \otimes_{t=1}^{T}\Pi_t$, where $\Pi_t = \{ \{R_{t\ell}, a_{t\ell}\}_{\ell=1}^{L_t}:  
R_{t\ell} \in \mathcal{R}_t, a_{t\ell} \in \mathcal{A}_t, 
L_t \leq L_{\text{max}} \}$.

\begin{remark}
We omit sets of the form  
$\{\bm{x}_t \in \mathcal{X}_t: x_{j_1} \leq \tau_1 \text{ or } x_{j_2} \leq \tau_2 \}$
in the definition of $\mathcal{R}_t$
because such sets are expressible in terms of the sets already in $\mathcal{R}_t$.
For example, the clause ``if $\bm{x}_t \in R_{t1}$ then $a_{t1}$''
with $R_{t1} = \{\bm{x}_t \in \mathcal{X}_t: x_{tj_1} \leq \tau_1 \text{ or } x_{tj_2} \leq \tau_2 \}$
can be written as
``if $\bm{x}_t \in R'_{t1}$ then $a_{t1}$; else if $\bm{x}_t \in R'_{t2}$ then $a_{t1}$''
with $R'_{t1} = \{\bm{x}_t \in \mathcal{X}_t: x_{tj_1} \leq \tau_1\}$
and $R'_{t2} = \{\bm{x}_t \in \mathcal{X}_t: x_{tj_2} \leq \tau_2\}$.
Moreover, the latter form has the benefit of avoiding the measurement of $x_{j_2}$
for subjects satisfying $x_{j_1} \leq \tau_1$, which may be an important consideration
if $x_{j_2}$ refers to some biomarker that is expensive to measure 
\citep[see][for discussion of decision lists and measurement cost]{zhang2015using}.
\end{remark}

\begin{remark}
  Under certain generative models, distinct sets in $\mathcal{R}_t$
  may correspond to the same group of subjects with probability one.
  For example, if $X_{t1}$ takes values in $\{0, 1\}$, the set
  $\{\bm{x} \in \mathcal{X}_t: x_1 \leq 0\}$ and the set
  $\{\bm{x} \in \mathcal{X}_t: x_1 \leq 0.5\}$ correspond to the
  same group of subjects.  To address this issue, it is tempting to
  require the threshold for $x_1$ to take values in the support of
  $X_{t1}$.  Nevertheless, such requirement is not sufficient to
  ensure that different sets in $\mathcal{R}_t$ correspond to
  different groups of subjects.  To see this, suppose
  $(X_{t1}, X_{t2})^\T$ can take three possible values:
  $(0, 0)^\T$, $(1, 0)^\T$ and $(1, 1)^\T$, e.g., if $X_{t1}$ and
  $X_{t2}$ are indicators of two symptoms where the second symptom can
  be present only when the first symptom is present.  In this case, the set
  $\{\bm{x} \in \mathcal{X}_t: x_1 \leq 0\}$ and the set
  $\{\bm{x} \in \mathcal{X}_t: x_1 \leq 0 \text{ and } x_2 \leq 0\}$
 correspond to the same group of subjects.  Therefore, we allow
  the thresholds to take arbitrary values.   In our theoretical
  analysis, we quantify dissimilarity of sets in $\mathcal{R}_t$
  using a distance that accounts for the distribution of
  $\bm{X}_{t}$.
\end{remark}


To estimate $\bm{\pi}^{\mathrm{opt}}$ we combine non-parametric
$Q$-learning with policy-search \citep[see][for a discussion of this
idea in the context of single decision point]{taylor2015reader}.  
To develop our ideas, we first
provide a high-level schematic for our algorithm, then we describe
implementation and modeling details, and finally we 
discuss a computational insight that improves computation time. 
A complete description of our estimation algorithm is lengthy and
technical  and is therefore presented in the Supplemental Materials. 

Define
$Q_T(\bm{x}_T, a_T) = \mathbb{E}(Y_T|\bm{X}_T=\bm{x}_T, A_T =a_T)$.
Then it can be shown that
$\pi_{T}^{\mathrm{opt}} = \arg\max_{\pi \in
  \Pi_T}\mathbb{E}Q_T\left\lbrace \bm{X}_T, \pi\left( \bm{X}_T \right)
\right\rbrace$.
Recursively, for $t=T-1,\ldots, 1$ define
$Q_t(\bm{x}_{t}, a_t) = \mathbb{E}\left[ Y_t + Q_{t+1}\left\lbrace
    \bm{X}_{t+1}, \pi_{t+1}^{\mathrm{opt}}\left(\bm{X}_{t+1}\right)
  \right\rbrace \big| \bm{X}_{t}=\bm{x}_t, A_t=a_t \right]$
and subsequently it can be shown that
$\pi_{t}^{\mathrm{opt}} = \allowbreak
\arg\max_{\pi_t\in\Pi_t}\mathbb{E} Q_{t}\left\lbrace \bm{X}_t,
  \pi_t(\bm{X}_t)\right\rbrace$
\citep[][]{schulte2014q-and-a}.  For each $t=1,\ldots, T$ let
$\mathcal{Q}_t$ denote a postulated class of models for $Q_t$.  
$Q$-learning with policy-search follows directly from the foregoing
definitions, a schematic is as follows.

\begin{itemize}
  \item[(S1)]  Construct an estimator of $Q_T$ in $\mathcal{Q}_T$,
    e.g.,  one could use penalized least squares   
$\widehat{Q}_{T} =
\arg\min_{Q_T\in\mathcal{Q}_T}\sum_{i=1}^{n}\left\lbrace
Y_{Ti} - Q_{T}(\bm{X}_{Ti},. A_{Ti})
\right\rbrace^2 + \mathcal{P}_T(Q_T)$, where $\mathcal{P}_T(Q_T)$ is a penalty
on the complexity of $Q_T$.   Define $\widehat{\pi}_{T} =
\arg\max_{\pi\in\Pi_T} \sum_{i=1}^{n}\widehat{Q}_{T}\left\lbrace
\bm{X}_{Ti}, \pi_T(\bm{X}_{Ti})
\right\rbrace$.  
\item[(S2)] Recursively, for $t=T-1, \ldots, 1$ construct an estimator
of $Q_t$ in $\mathcal{Q}_t$, say $\widehat{Q}_t$, e.g.,  
$$\widehat{Q}_{t} =
\arg\min_{Q_t\in\mathcal{Q}_t}\sum_{i=1}^{n}\left\lbrace
Y_{ti} + \widehat{Q}_{t+1}\left\lbrace
\bm{X}_{(t+1)i}, \widehat{\pi}_{t+1}(\bm{X}_{(t+1)i})
\right\rbrace
 - Q_{t}(\bm{X}_{ti},. A_{ti})
\right\rbrace^2 + \mathcal{P}_t(Q_t),$$ where $\mathcal{P}_t(Q_t)$ is a penalty
on the complexity of $Q_t$.  Define $\widehat{\pi}_t =\\
\arg\max_{\pi_t\in\Pi_t} \sum_{i=1}^{n}\widehat{Q}_{t}\left\lbrace
\bm{X}_{ti}, \pi_t(\bm{X}_{ti})
\right\rbrace$. 
\end{itemize}


Implementation of the preceding schematic requires a choice of models
for the $Q$-functions, a means of constructing an estimator within
this class, and an algorithm for computing
$\arg\max_{\pi_t\in\Pi_t}\sum_{i=1}^{n}\widehat{Q}_{t}\left\lbrace
  \bm{X}_{ti}, \pi_t\left( \bm{X}_{ti} \right) \right\rbrace$.
In our implementation, we use kernel ridge regression with an extended
Gaussian kernel to construct estimators of the $Q$-functions and a
greedy stepwise algorithm to approximate $\widehat{\pi}_{t}$ from the estimated
$Q$-functions.


\subsection{Kernel Ridge Regression}
\label{subsec:krr}

We use kernel ridge regression to estimate the $Q$-functions.
Starting with the last stage, let $K_T(\cdot, \cdot)$ be a symmetric
and positive definite function from
$\mathbb{R}^{d_T} \times \mathbb{R}^{d_T}$ to $\mathbb{R}$, and 
let $\mathbb{H}_T$ be the corresponding
reproducing kernel Hilbert space (RKHS).  
In our implementation,
we employ an extension of the
Gaussian kernel that employs different scaling factors in different
variables:
$ K_T(\bm{x}, \bm{z}) = \exp\left\{ -\sum_{j=1}^{d_T} \gamma_{Tj} (x_j
  - z_j)^2 \right\} $,
where $\bm{\gamma}_T = (\gamma_{T1}, \dots, \gamma_{T d_T})^\T$ is a
tuning parameter and $\gamma_{Tj} > 0$ for all $j$.   
For each $a \in \mathcal{A}_T$, 
we estimate $Q_T(\cdot, a)$ via penalized least squares
\[
\wh{Q}_T(\cdot, a) = \argmin_{f \in \mathbb{H}_T} 
\frac{1}{n_{Ta}} \sum_{i \in \mathcal{I}_{Ta}} \{Y_{iT} - f(\bm{X}_{iT})\}^2
+ \lambda_T \lVert f \rVert_{\mathbb{H}_T}^2,
\]
where $\mathcal{I}_{Ta} = \{i: A_{iT} = a\}$, $n_{Ta} =
|\mathcal{I}_{Ta}|$,
and $\lambda_{T}> 0$ is a tuning parameter.  
Let
$\bm{Y}_{Ta} = (Y_{iT})_{i \in \mathcal{I}_{Ta}}$ and
$\bm{K}_{Ta} = \{K(\bm{X}_{iT}, \bm{X}_{jT})\}_{i,j \in \mathcal{I}_{Ta}}$.
By the representer theorem \citep{kimeldorf1971some}, 
$\wh{Q}_T(\bm{x},a) = 
\sum_{i \in \mathcal{I}_{Ta}} K_T(\bm{x}, \bm{X}_{iT}) \wh{\beta}_{iTa}$,
where
$\wh{\bm{\beta}}_{Ta} = (\wh{\beta}_{iTa})_{i \in \mathcal{I}_{Ta}}$ satisfy
$\wh{\bm{\beta}}_{Ta} = \argmin_{\bm{\beta}} 
\lVert \bm{Y}_{Ta} - \bm{K}_{Ta} \bm{\beta} \rVert^2
+ n_{Ta} \lambda_T \bm{\beta}^\T \bm{K}_{Ta} \bm{\beta}$.  
Define $\widehat{\pi}_{T} =\arg\max_{\pi_T\in\Pi_T}\sum_{i=1}^{n}
\widehat{Q}_{T}\left\lbrace \bm{X}_{Ti},
  \pi_{T}(\bm{X}_{Ti})\right\rbrace$.  

Similarly, for each $t < T$ let 
$\mathbb{H}_t$ be the RKHS induced by the kernel
$ K_t(\bm{x}, \bm{z}) = \exp\left\{
-\sum_{j=1}^{d_t} \gamma_{tj} (x_j - z_j)^2
\right\} $,
and
$\bm{\gamma}_t = (\gamma_{t1}, \dots, \gamma_{t d_t})^\T$
is a tuning parameter.
Recursively, for each $t < T$, $a_t \in \mathcal{A}_t$, 
estimate $Q_t(\cdot, a)$ by
\[
\wh{Q}_t(\cdot, a) = \argmin_{f \in \mathbb{H}_t} 
\frac{1}{n_{ta}} \sum_{i \in \mathcal{I}_{ta}} \left[
Y_{it} + \wh{Q}_{t+1}\{\bm{X}_{i,t+1}, \wh{\pi}_{t+1}(\bm{X}_{i,t+1})\} - f(\bm{X}_{it}) \right]^2 
+ \lambda_t \lVert f \rVert_{\mathbb{H}_t}^2,
\]
where 
$\mathcal{I}_{ta} = \{i: A_{it} = a\}$,
$n_{ta} = |\mathcal{I}_{ta}|$,
$\mathbb{H}_t$ is an RKHS induced by the kernel
$ K_t(\bm{x}, \bm{z}) = \exp\left\{
-\sum_{j=1}^{d_t} \gamma_{tj} (x_j - z_j)^2
\right\} $,
and $\lambda_t$,
$\bm{\gamma}_t = (\gamma_{t1}, \dots, \gamma_{t d_t})^\T$
are tuning parameters.

\subsection{Construction of Decision Lists}
\label{subsec:dl}
In addition to a method for estimating the $Q$-functions, the proposed
method requires a method for computing
$\arg\max_{\pi_t\in\Pi_t}\sum_{i=1}^{n} \widehat{Q}_{t}\left\lbrace
  \bm{X}_{ti}, \widehat{\pi}_t(\bm{X}_{ti}) \right\rbrace$
where $\Pi_t$ is the space of list-based decision rules defined
previously.  Any element in $\Pi_t$ can be expressed as
$\{(R_{t\ell}, a_{t\ell})\}_{\ell=1}^{L_t}$,  however, simultaneous
optimization over all regions and treatments is not computationally 
feasible except in very small problems.  Instead, we propose an
algorithm that constructs $\widehat{\pi}_{t}$ using a greedy 
optimization procedure that optimizes one clause in
$\widehat{\pi}_{t}$ at a time; unlike many greed algorithms,  the
proposed method is consistent for the global maximizer.  
To provide intuition, we describe in detail the first two steps of this greedy
algorithm before stating it in more general terms.  

\subsubsection{Estimation of the first clause}
Define $\widehat{\pi}_{t}^{Q}$ to be map 
$\bm{x}_{t} \mapsto
\arg\max_{a_t\mathcal{A}_t}\widehat{Q}_{t}(\bm{x}_{t}, a_t)$; thus,
$\widehat{\pi}_{t}^{Q}$ is an optimal estimated decision rule at stage
$t$ using non-parametric $Q$-learning.   
To estimate the first clause $(R_{t1}, a_{t1})$ in $\pi_t$,
we consider the following decision-list parameterized by $R$ and $a$:
\begin{align}
\label{eq:algo-regime1}
&\texttt{If } \bm{x}_t \in R \texttt{ then } a; \nonumber \\
&\texttt{else if } \bm{x}_t \in \mathcal{X}_t \texttt{ then } \wh{\pi}_t^Q(\bm{x}_t).
\end{align}
If all subjects follow \eqref{eq:algo-regime1}, the estimated mean outcome is
\begin{equation}
\label{eq:algo-outcome1}
\frac{1}{n} \sum_{i=1}^n \left[ I(\bm{X}_{it} \in R) \wh{Q}_t(\bm{X}_{it}, a)
  + I(\bm{X}_{it} \notin R) \wh{Q}_t\{\bm{X}_{it}, \wh{\pi}_t^Q(\bm{X}_{it})\} \right].
\end{equation}
Hence, we can pick the maximizer of \eqref{eq:algo-outcome1} as the
estimator of $(R_{t1}, a_{t1})$.  Note that the difference between the
estimated mean outcome under $\wh{\pi}_t^Q$ and that under
\eqref{eq:algo-regime1} is
$ n^{-1} \sum_{i=1}^n I(\bm{X}_{it} \in R) \left[
  \wh{Q}_t\{\bm{X}_{it}, \wh{\pi}_t^Q(\bm{X}_{it})\} -
  \wh{Q}_t(\bm{X}_{it}, a) \right] $,
which measures the decrease in the estimated mean outcome when some
part of $\wh{\pi}_t^Q$ is replaced with an if-then clause.  This represents
the price paid for interpretability, and by maximizing
\eqref{eq:algo-outcome1}, we minimize this price.  

To improve generalization performance, we add a complexity penalty
to \eqref{eq:algo-regime1}; in addition to encouraging parsimonious
lists,  we shall see that this penalty also ensures a unique
maximizer.  Define $V(R) \in \left\lbrace 0,1,3\right\rbrace$ 
to be the number of covariates needed to check inclusion in $R$.  
We define $\widehat{R}_1$ and $\widehat{a}_1$ as the maximizers over
$R$ and $a$ in 
\begin{multline}
\label{eq:algo-outcome1a}
\frac{1}{n} \sum_{i=1}^n \left[ I(\bm{X}_{it} \in R) \wh{Q}_t(\bm{X}_{it}, a)
+ I(\bm{X}_{it} \notin R) \wh{Q}_t\{\bm{X}_{it}, \wh{\pi}_t^Q(\bm{X}_{it})\} \right]\\
{} + \zeta \left\{ \frac{1}{n} \sum_{i=1}^n I(\bm{X}_{it} \in R) \right\} + \eta \{ 2 - V(R) \},
\end{multline}
where $\zeta, \eta > 0$ are tuning parameters.  Thus, the first
penalty term rewards regions $R$ with large mass relative to the 
distribution of $\bm{X}_{t}$ whereas the second term rewards
regions that involve fewer covariates.   
Moreover, we impose the constraint 
$n^{-1} \sum_{i=1}^n I(\bm{X}_{it} \in R) > 0$
to avoid searching over vacuous clauses.

\subsubsection{Estimation of the second clause}
To estimate the second clause
we consider the following decision list parameterized by $R$ and $a$
\begin{align}
\label{eq:algo-regime2}
&\texttt{If } \bm{x} \in \wh{R}_{t1} \texttt{ then } \wh{a}_{t1}; \nonumber \\
&\texttt{else if } \bm{x} \in R \texttt{ then } a; \nonumber \\
&\texttt{else if } \bm{x} \in \mathcal{X}_t \texttt{ then } \wh{\pi}_t^Q(\bm{x}).
\end{align}
If all the subjects follow the regime \eqref{eq:algo-regime2}, the
estimated mean outcome is
\begin{multline}
\label{eq:algo-outcome2}
\frac{1}{n} \sum_{i=1}^n I(\bm{X}_{it} \in \wh{R}_{t1}) 
\wh{Q}_t(\bm{X}_{it}, \wh{a}_{t1}) + 
\frac{1}{n} \sum_{i=1}^n I(\bm{X}_{it} \notin \wh{R}_{t1}, \bm{X}_{it} \in R) 
\wh{Q}_t(\bm{X}_{it}, a) \\
{} + \frac{1}{n} \sum_{i=1}^n I(\bm{X}_{it} \notin \wh{R}_{t1}, \bm{X}_{it} \notin R) 
\wh{Q}_t\{\bm{X}_{it}, \wh{\pi}_t^Q(\bm{X}_{it})\}.
\end{multline}
Note that the first term in \eqref{eq:algo-outcome2} can be dropped
during the optimization as it is independent of $R$ and $a$.  As in
\eqref{eq:algo-outcome1a},  we maximize the penalized criterion
\begin{multline}
\label{eq:algo-outcome2a}
\frac{1}{n} \sum_{i=1}^n I(\bm{X}_{it} \notin \wh{R}_{t1}, \bm{X}_{it} \in R) 
\wh{Q}_t(\bm{X}_{it}, a)
+ \frac{1}{n} \sum_{i=1}^n I(\bm{X}_{it} \notin \wh{R}_{t1}, \bm{X}_{it} \notin R) 
\wh{Q}_t\{\bm{X}_{it}, \wh{\pi}_t^Q(\bm{X}_{it})\} \\
{} + \zeta \left\{ \frac{1}{n} \sum_{i=1}^n I(\bm{X}_{it} \notin \wh{R}_{t1}, \bm{X}_{it} \in R) \right\} 
+ \eta \{ 2 - V(R) \}
\end{multline}
with respect to $R \in \mathcal{R}_t, a \in \mathcal{A}_t$ and subject
to the constraint
$n^{-1} \sum_{i=1}^n I(\bm{X}_{it} \notin \wh{R}_{t1}, \bm{X}_{it} \in
R) > 0$.
We continue this procedure until either every subject gets a
recommended treatment, namely $R_{t\ell} = \mathcal{X}_t$ for some
$\ell$, or the maximum length is reached, 
$\ell = L_{\text{max}}$.  If the maximum list length is reached, 
we set
$R_{tL_{\max}} = \mathcal{X}_t$ to ensure that the regime applies to every
subject and choose $\widehat{a}_{tL_{\max}}$ be the estimated 
best single treatment for all remaining subjects.

\subsubsection{Estimation of all clauses}

An algorithmic description of the proposed algorithm is given below.
Additional
computational details, including the time complexity, are given
in the next section.  
\begin{enumerate}
\item[]Step~1. Initialize $\ell = 1$.

\item[]Step~2. If $\ell < L_{\text{max}}$, compute
\begin{multline}
\label{eq:sample-objective}
(\wh{R}_{t\ell}, \wh{a}_{t\ell}) = 
\argmax_{R \in \mathcal{R}_t , a \in \mathcal{A}_t}\;
\frac{1}{n} \sum_{i=1}^n \bigg[
 I(\bm{X}_{it} \in \wh{G}_{t\ell}, \bm{X}_{it} \in R) \wh{Q}_t(\bm{X}_{it}, a)  \\
 {} + I(\bm{X}_{it} \in \wh{G}_{t\ell}, \bm{X}_{it} \notin R) \wh{Q}\{\bm{X}_{it}, \wh{\pi}_t^Q(\bm{X}_{it})\}
\bigg]  \\
 {} + \zeta \bigg\{ \frac{1}{n} \sum_{i=1}^n I(\bm{X}_{it} \in \wh{G}_
{t\ell}, \bm{X}_{it} \in R) \bigg\}
+ \eta \{2 - V(R)\}
\end{multline}
subject to $n^{-1} I(\bm{X}_{it} \in \wh{G}_
{t\ell}, \bm{X}_{it} \in R) > 0$,
where 
$\wh{G}_{t1} = \mathcal{X}_t$,
$\wh{G}_{t\ell} = \mathcal{X}_t \setminus
\big( \bigcup_{k < \ell} \wh{R}_{tk} \big)$ for $\ell \geq 2$,
and $V(R) \in \{0, 1, 2\}$ 
is the number of variables used to define $R$.
It is easy to verify that the objective function above
reduces to \eqref{eq:algo-outcome1a} when $\ell=1$
and to \eqref{eq:algo-outcome2a} when $\ell=2$.
If $\ell = L_{\text{max}}$, set
\begin{equation}
\label{eq:sample-objective2}
(\wh{R}_{t\ell}, \wh{a}_{t\ell}) = 
\argmax_{R \in \mathcal{R}_t, a \in \mathcal{A}_t}
\frac{1}{n} \sum_{i=1}^n I(\bm{X}_{it} \in \wh{G}_{t\ell}) \wh{Q}(\bm{X}_{it}, a) + \eta\{ 2 - V(R) \}.
\end{equation}
The solution of \eqref{eq:sample-objective2}
must satisfy $V(R) = 0$ and hence $\wh{R}_{t\ell} = \mathcal{X}_t$.
Consequently the last clause does apply to all the rest subjects.

\item[]Step~3. If $\wh{R}_{t\ell} = \mathcal{X}_t$ then go to Step~4;
otherwise, increase $\ell$ by $1$ and repeat Steps~2~and~3.

\item[]Step~4. Output $\wh{\pi}_t = \{(\wh{R}_{tk}, \wh{a}_{tk})\}_{k=1}^{\ell}$.
\end{enumerate}

\subsubsection{Implementation details and time complexity}
Computation of $(\wh{R}_{t\ell}, \wh{a}_{t\ell})$ in
\eqref{eq:sample-objective} requires special attention because the
objective function is non-differentiable and non-convex.  We first
argue that brute-force search can be used to obtain
$(\wh{R}_{t\ell}, \wh{a}_{t\ell})$.  Although $\mathcal{R}_t$ contains
infinitely many elements, because the objective function in
\eqref{eq:sample-objective} is piecewise linear, for each covariate it
suffices to consider $n$ thresholds located at the order statistics of
that covariate.  Hence, the number of thresholds to enumerate is of
order $n^2$.  In addition, there are $d_t (d_t+1) / 2$ choices for
variables in $R_{t\ell}$, $m_t$ choices for $a_{t\ell}$, and each
evaluation of \eqref{eq:sample-objective} takes $O(n)$ operations.
Therefore, the time complexity for finding
$(\wh{R}_{t\ell}, \wh{a}_{t\ell})$ via brute-force search is
$O(n^3 d_t^2 m_t)$.  Unfortunately, the factor $n^3$ is overwhelming
even when the sample size $n$ is moderate.

Instead of brute-force search, we propose a novel algorithm
to compute $(\wh{R}_{t\ell}, \wh{a}_{t\ell})$, 
that substantially reduces the time complexity.
Note that the $n^3$ factor is due to the enumeration of thresholds and the evaluation of the objective function in \eqref{eq:sample-objective}.
By reorganizing the enumeration and evaluation,
the proposed algorithm reduces the $n^3$ factor to $n \log n$. Thus,
with this implementation, the proposed algorithm can be applied to
large datasets; this is appealing in an era of `big-data' where large 
data-bases are being mined to generate hypotheses about precision
medicine.   
\begin{proposition}
\label{thm:time-comp}
For each $t$ and $\ell$, 
the estimator $(\wh{R}_{t\ell}, \wh{a}_{t\ell})$ in \eqref{eq:sample-objective} 
can be computed within $O(n \log n \, d_t^2 m_t)$ operations.
\end{proposition}
The proof of this result is constructive but technical so we provide a
sketch of the main idea here and relegate the remaining details to the
Supplemental Materials.  Suppose $R$ involves only one covariate:
$R = \{\bm{x}: x_j \leq \tau\}$.  For fixed $t$, $j$ and $a$, we
observe that, up to a constant independent of $\tau$, the objective
function in \eqref{eq:sample-objective} is of the form
$F(\tau) = n^{-1} \sum_{i=1}^n I(X_{ijt} \leq \tau) U_i + I(X_{ijt} >
\tau) V_i$,
where $U_i$ and $V_i$ are constants.  As discussed previously, we
need only to compute $F(\tau)$ for $\tau$ equal to observed covariate values, $X_{ijt}$.
Let $i_1 < \dots < i_n$ be a permutation of $1, \dots, n$ such that
$X_{i_1 jt} \leq \dots \leq X_{i_n jt}$.  Then, it can be shown that
$F(X_{i_s jt}) = F(X_{i_{s-1} jt}) + U_{i_s} - V_{i_s}$, $s \geq 2$.
Hence, one can enumerate all possible values for $\tau$ and evaluate
$F(\tau)$ in $O(n)$ time, in contrast to $O(n^2)$ time for brute-force
search.  A similar recursive relationship can be established if $R$ is
of the form $\{\bm{x}: x_j > \tau\}$.  When $R$ involves two
covariates, we combine this sorting technique with binary search tree
\citep{cormen2009introduction}, which enables us to find the
thresholds in $O(n \log n)$ time.

\begin{remark}
The proposed algorithm differs from that in \citet{zhang2015using}
in two important ways.
First, the two algorithms maximize different objective functions.
In 
\citet{zhang2015using},
regime \eqref{eq:algo-regime1} is replaced by
``if $\bm{x} \in R$ then $a$; else if $\bm{x} \in \mathcal{X}_t$ then $a'$'',
where $R$, $a$ and $a'$ are obtained by maximizing the estimated mean outcome
under such a regime.
However, this criterion fails to account for subsequent splits in the
decision lists and can thereby get stuck in  a local mode. In
contrast, the proposed algorithm approximates the remaining list with
the estimated optimal regime using non-parametric $Q$-learning.  
To illustrate the difference between the two objective functions,
consider a scenario with $T=1$ stage, 
a single covariate $S_1 \sim \text{Uniform}(-2, 2)$ 
and suppose that $\wh{Q}_1(x, a) = Q_1(x, a) = ax(x-1), a \in \{-1, 1\}$.
Assume $\zeta$ and $\eta$ are small but positive.
Then the solution of \eqref{eq:sample-objective} is
$\wh{R}_{11} = \{x: x \leq \tau\}$ and $\wh{a}_{11} = 1$
with $\tau \approx 0$. 
Nevertheless, if the term $\wh{\pi}^Q_t(\bm{X}_{it})$
were replaced by a fixed treatment $a' \neq a$,
the solution would be $\wh{R}_{11} = \mathcal{X}_1$ and $\wh{a}_{11}=1$,
leading to a suboptimal  regime. 
A second difference between the proposed algorithm and the one proposed 
in \citet{zhang2015using} is that the latter requires a pre-specified
set of candidate thresholds for each predictor,
and its time complexity is the same as brute-force search
if we use all the unique values as candidate thresholds. 
\end{remark}

\section{Theoretical Results}
\label{sec:theory}
For each $\ell = 1, 2, \dots$, define the population analogs of 
\eqref{eq:sample-objective} and \eqref{eq:sample-objective2} 
as follows,
$(R^{\ast}_{t\ell}, a^{\ast}_{t\ell}) = \argmax_{R \in \mathcal{R}_t,
  a \in \mathcal{A}_t} \Psi_{t\ell}(R, a)$, where
\begin{multline}
\label{eq:population-objective}
\Psi_{t\ell}(R,a) = E\left[ I(\bm{X}_t \in G^{\ast}_{t\ell}, \bm{X}_t
  \in R) Q(X_t, a) + I(\bm{X}_t \in G^{\ast}_{t\ell}, \bm{X}_t \notin
  R) Q\left\{\bm{X}_t, \pi^{Q}_{t}(\bm{X}_t)\right\} \right] \\ +
\zeta \Pr(\bm{X}_t \in G^{\ast}_{t\ell}, \bm{X}_t \in R) + \eta
\{2-V(R)\},
\end{multline}
and
$G^{\ast}_{t\ell} = \mathcal{X}_t$ if $\ell=1$ 
and $G_{t\ell}^{\ast} = \mathcal{X}_t \setminus (\cup_{k<\ell} R^{\ast}_{t k})$ otherwise,
until either $R^{\ast}_{t\ell} = \mathcal{X}_t$ or $\ell = L_{\text{max}}$.
In the latter case, instead of \eqref{eq:population-objective} we define
\begin{equation}
\label{eq:population-objective2}
\Psi_{t\ell}(R, a) = E \left\{ I(X_t \in G^{\ast}_{t\ell}) Q(\bm{X}_t, a) \right\}
+ \eta\{2-V(R)\}.
\end{equation}
Let $L_t^* = \min\{\ell: R_{t\ell}^* = \mathcal{X}_t\}$ and
$\pi_t^* = \big\{ (R^{\ast}_{t\ell}, a^{\ast}_{t\ell}) \big\}_{\ell =
  1}^{L_t^*}$.
In \eqref{eq:population-objective} and
\eqref{eq:population-objective2}, the $Q$-functions are defined as
$Q_T(\bm{x}, a) = E(Y_T | \bm{X}_T=\bm{x}, A_T=a)$,
$Q_t(\bm{x}, a) = E[Y_t + Q_{t+1} \{ \bm{X}_{t+1}, \pi^{\ast}_{t+1}
(\bm{X}_{t+1})\} | \bm{X}_t= \bm{x}, A_t=a ]$
for $t=T-1, \dots, 1$. Furthermore, let
$\pi_{t}^{Q}(\bm{x}) = \argmax_{a \in \mathcal{A}_t} Q_t(\bm{x}, a)$
for all $t$.

We assume that all the covariates and outcomes are bounded. 
This is a common assumption in the context  of nonparametric
regression; the extension to include unbounded covariates is possible but at the
expense of additional complexity.  
\begin{assumption}
There exists $b>0$ such that $\lVert \bm{X}_t \rVert_\infty \leq b$ 
and $|Y_t| \leq b$ with probability one for all $t=1, \dots, T$. 
\end{assumption}
We also assume positivity \cite[][]{robins2004optimal}, which ensures that
$Q_t(\bm{x}, a)$ is well-defined for all $a\in
\mathcal{A}_t$.
\begin{assumption}
For each $t$ and $a \in \mathcal{A}_t$, $\pr(A_t = a | \bm{X}_t) \geq
\varpi$ almost surely for some positive constant
$\varpi$.\end{assumption}

A crucial intermediate step in deriving the asymptotic behavior of
$\wh{\pi}_t$'s is establishing convergence of $\wh{Q}_t$ to $Q_t$; to
facilitate this step we require a certain degree of smoothness in
$Q_t$.
A common means of imposing smoothness is to assume differentiability
\citep[see, e.g.,][]{stone1982optimal}.
However, the non-differentiable maximization operator that is implicit 
in the definition of the $Q$-functions forces us to consider a weaker
notion of smoothness. 
 Denote $B_t =
[-b, b]^{d_t} \subset \mathbb{R}^{d_t}$.  For any function $f: B_t \to
\mathbb{R}$, define the $r$-th difference $\Delta^{r}_{\bm{h}}(f)$ by
$\Delta^{r}_{\bm{h}}(f)(\bm{x}) = \sum^{r}_{i=0}
\binom{r}{i}(-1)^{r-i}f(\bm{x} + i\bm{h})$ if $\bm{x} \in B_{t, r,
  \bm{h}}$ and $0$ otherwise, where $r$ is a positive integer,
$\bm{h}=(h_1, \dots, h_{d_t})^T$, $h_j \geq 0$ for all $j=1, \dots,
d_t$, and $B_{t, r, \bm{h}} = \{ \bm{x} \in B_t: \bm{x} + i \bm{h} \in
B_t \text{ for all } i \leq r \}$.  Define the $r$-th modulus of
smoothness of $f$ by $\omega_r(f, s) = \sup_{\lVert \bm{h} \rVert_2
  \leq s} \sup_{\bm{x} \in B_t} \lvert \Delta^{r}_{\bm{h}}(f)(\bm{x})
\rvert$.  The definition above is similar to \citet[Definition
  2.1]{eberts2013optimal}, but replaces the $L_p$ norm with the
supremum norm.  This modification allows us to drop the requirement
that $\bm{X}_t$ have a density with respect to  Lebesgue measure.
Thus, our analysis applies when $\bm{X}_t$ contains discrete
covariates.

The concept of modulus of smoothness generalizes the concept of
differentiability.  To see this, consider an example where $d=1$. We
observe that $\lim_{h \to 0} h^{-1}\Delta_{h}^{1}(f) (x) =
f'(x)$.
Suppose $|f'(x)|$ is bounded, then for sufficiently small $h$, there
exists a constant $C_f$ such that
$|\Delta^{1}_{h}(f)(x)| \leq C_f|h|$. Hence, any continuously
differentiable function $f$, defined on a finite interval, satisfies
$\omega_1(f, s)= O(s)$, as $s \to 0$.  Generally, if $f$ is $r$-times
continuously differentiable, then $w_r(f, s) = O(s^r)$ as $s \to 0$.
In addition, some non-differentiable functions also satisfy this
condition. Consider $f(x) = |x|$ and $f(x) = \max(x, 0)$.  It is easy
to verify that $|\Delta_{h}^{1}(f)(x)| \leq |h|$ for any $x$.  Thus
$f(x)=|x|$ and $f(x)=\max(x,0)$ also satisfy $\omega_1(f, s) = O(s)$
though $f$ is not differentiable at $0$.  We make the following
assumption regarding the smoothness of the $Q$-functions.

\begin{assumption}
For each $t$, there exists a positive integer $r_t$ such that 
$\omega_r\{Q^{\ast}_t(\cdot, a), s\} = O(s^{r_t})$ as $s \to 0$,
for any $a \in \mathcal{A}_t$. 
\end{assumption}

In order to study the probabilistic convergence of $\wh{\pi}_t$ to
$\pi^{\ast}_{t}$, it is necessary to define an appropriate distance
between $\wh{R}_{t\ell}$ and $R^{\ast}_{t\ell}$.  In view of Remark~2,
the distance should incorporate the distribution of
$\bm{X}_t$, thus, we  define
$\rho_t(R_1, R_2) = \Pr \{ \bm{X}_t \in (R_1 \symmdiff R_2) \} $,
where
$C\symmdiff D$ denotes the symmetric set difference between sets $C$ and $D$.  It
can be verified that $\rho_t(R_1, R_2)$ is non-negative, symmetric, and
satisfies the triangle inequality.  Note that $\rho_t(R_1, R_2)=0$
indicates only that $R_1$ and $R_2$ refer to the same group of
subjects with probability one  with respect to $\bm{X}_t$ but does not 
imply $R_1=R_2$.  For example, suppose $X_{t\ell}$ takes values in
$\{0, 1\}$, $R_1=\{ x: x \leq 0\}$ and $R_2=\{ x: x \leq 0.5\}$.  Then
$\rho_t(R_1, R_2)=0$, as expected.  Furthermore, the use of
$\rho_t$ helps to avoid the issue of non-unique representations of $R$
when some covariates can be expressed using others.  For example, if
$X_{t1} = -X_{t2}$ and both are continuous, then
$\rho_t(\{\bm{x}: x_1 \leq \tau\}, \{\bm{x}: x_2 > \tau\}) = 0$.  
Thus, our goal is to identify an equivalence class of clauses that
each describe the same subset of patients.  
We require the following identifiability assumption on the equivalence
class of optimal clauses.  
\begin{assumption}
For each $t$ and $\ell$, the following inequalities hold:
\begin{itemize}
\item[(\romanNum{1})] There exists a constant $\kappa > 0$ such that
$\Psi_{t\ell}(R, a^{\ast}_{t\ell}) 
\leq \Psi_{t\ell}(R^{\ast}_{t\ell} , a^{\ast}_{t\ell}) 
- \kappa \rho^{2}_{t}(R, R^{\ast}_{t\ell})$
as $\rho_{t}(R, R^{\ast}_{t\ell}) \to 0$;

\item[(\romanNum{2})] For any $\delta > 0$, 
there exists a constant $\epsilon > 0$ such that 
$\Psi_{t\ell}(R, a^{\ast}_{t\ell}) 
\leq \Psi_{t\ell}(R^{\ast}_{t\ell} , a^{\ast}_{t\ell}) 
- \epsilon$ 
for all $R \in \mathcal{R}_t$ 
with $\rho_{t}(R, R^{\ast}_{t\ell}) > \delta$;

\item[(\romanNum{3})] There exists a constant $\varsigma > 0$ 
such that 
$\Psi_{t\ell}(R, a) 
\leq \Psi_{t\ell}(R^{\ast}_{t\ell}, a^{\ast}_{t\ell}) 
- \varsigma$
for all $R \in \mathcal{R}_t$ and 
$a \in \mathcal{A}_t \setminus \{a^{\ast}_{t\ell}\}$. 
\end{itemize}
\end{assumption}

Assumption~4 guarantees the uniqueness of
$(R^{\ast}_{t\ell}, a^{\ast}_{t\ell})$ in the sense that if
$(\wt{R}^{\ast}_{t\ell}, \wt{a}^{\ast}_{t\ell})$ is another maximizer
of $\Psi_{t\ell}(R, a)$, then
$\rho(R^{\ast}_{t\ell}, \wt{R}^{\ast}_{t\ell}) = 0$ and
$a^{\ast}_{t\ell} = \wt{a}^{\ast}_{t\ell}$.  Moreover, condition
(\romanNum{1}) assumes that $\Psi(R, a^{\ast}_{t\ell})$ behaves like a
quadratic function in a neighborhood of $R^{\ast}_{t\ell}$.  When
$\bm{X}_t$ has bounded density and $R$, $R_{t\ell}^*$ use the same
covariates, $\Psi_{t\ell}$ can be viewed as a function of the
threshold values and condition (\romanNum{1}) implies that
$\Psi_{t\ell}$ behaves like a quadratic function near the optimal
threshold values, which is a common condition in parametric models.

Define the value of a decision rule at time $t$, say $\pi_t$,  as 
$V_t(\pi_t) = \E \big[ Q_t\{\bm{X}_t, \pi_t(\bm{X}_t) \} \big]$.
Our analysis focuses on how close $\wh{\pi}_t$ is to $\pi_t^*$, 
and how well $\wh{\pi}_t$ performs compared to $\pi_t^*$ in
terms of value.  
For each $t$, let $\mathcal{S}_t$ be the indices of the signal
variables defining the function $Q_t$, and let
$\mathcal{N}_t = \{1, \dots, d_t\} \setminus \mathcal{S}_t$ be the
indices of the noise variables.   Write
$d_t^{\mathcal{S}} = |\mathcal{S}_t|$,
$d_t^{\mathcal{N}} = |\mathcal{N}_t|$ and hence
$d_t = d_t^{\mathcal{S}} + d_t^{\mathcal{N}}$. 
Recall that $L_T^*$ is the length of
$\pi_t^*$.
 Define
$\phi_T = (2/3)^{L_T^*} (2 r_T) / (2 r_T + d_T)$, and
$\phi_t = (2/3)^{L_t^*} \min\{(2 r_t) / (2 r_t + d_t), \phi_{t+1}\}$
for $t = T-1, \dots, 1$.  Define $\phi_t^\mathcal{S}$ in the same way
but with $d_t$ replaced by $d_t^\mathcal{S}$.  As $\bm{\gamma}_t$,
$\lambda_t$, $\zeta$ and $\eta$ may depend on $n$, we may write
$\bm{\gamma}_{n,t}$, $\lambda_{n,t}$, $\zeta_n$ and $\eta_n$ to
emphasize such dependence.  The following theorem establishes finite
sample bounds.  A proof is given in the Supplemental Materials.
\begin{theorem}
For each $t$, assume 
$\max_j \gamma_{n, t, j} = O\{n^{2 / (2 r_t + d_t)}\}$,
$\lambda_t = O(n^{-1})$,
$\sup_n \zeta_n < \infty$, and
$\sup_n \eta_n < \infty$.
Under Assumptions 1-4, for any $\xi > 0$, 
\begin{align*}
&\pr\{\wh{\pi}_t(\bm{X}_t) \neq \pi_t^*(\bm{X}_t)\} 
\leq c_1 n^{-\phi_t + \xi}, \\
&\pr\{V_t(\pi_t^*) - V_t(\wh{\pi}_t) \geq 
c_2 n^{-\phi_t + \xi} + c_3 n^{-1/2} \tau\}
\leq e^{-\tau},
\end{align*}
where $c_i$'s are constants independent of $n$ and $\tau$.

Moreover, if $\max_{j \in \mathcal{S}_t} \gamma_{n, t, j} = O\{n^{2 / (2 r_t + d_t)}\}$ and $\max_{j \in \mathcal{N}_t} \gamma_{n, t, j} = O(1)$, then the inequalities above holds with $\phi_t$ replaced by $\phi_t^\mathcal{S}$.
\end{theorem}

The minimax convergence rate for a nonparametric regression 
estimator of an $r_T$-times continuously differentiable function is
$O\{n ^ {-(2 r_T) / (2 r_T + d_T)}\}$ \citep{stone1982optimal}.  By
extending the technique in \citet{eberts2013optimal}, we show in the 
Supplemental Material that the
estimated $Q$-function $\wh{Q}_T$ converges to its true value $Q_T$ at
a nearly optimal rate $O\{n ^ {-(2 r_T) / (2 r_T + d_T) + \xi'}\}$,
where $\xi' > 0$ can be arbitrarily small.  The construction of
$\wh{\pi}_T$ involves estimating $L_T^*$ pairs of parameters
$(R_{t\ell}, a_{t\ell})$, one pair for each if-then clause.  The
estimation of each pair $(R_{t\ell}, a_{t\ell})$ reduces the
convergence rate by an additional factor of $2/3$.  The underlying
idea for this phenomenon is analogous to the problems analyzed in
\citet{kim1990cube}.  In earlier stages, the estimation of
$Q$-functions is further complicated by the fact that
$Q_{t+1}\{\bm{X}_{t+1}, \pi_{t+1}(\bm{X}_{t+1})\}$ is not observed but
estimated via
$\wh{Q}_{t+1}\{\bm{X}_{t+1}, \wh{\pi}_{t+1}(\bm{X}_{t+1})\}$.

When all the covariates are discrete, 
it can be shown that
the convergence rate of the estimated regime $\wh{\pi}_t$
does not inherit the slow convergence rate from the 
underlying nonparametric regressions.  The following result is proved
in the Supplemental Material. 
\begin{theorem}
For each $t$, assume 
$\max_j \gamma_{n, t, j} = O\{n^{2 / (2 r_t + d_t)}\}$,
$\lambda_t = O(n^{-1})$,
$\sup_n \zeta_n < \infty$, and
$\sup_n \eta_n < \infty$.
Furthermore, assume that the distribution of $\bm{X}_t$ is discrete.
Namely, for each $t$ there exists a finite set $\wt{\mathcal{X}}_t$
such that $\pr(\bm{X}_t \in \wt{\mathcal{X}}_t) = 1$.
Under Assumptions 1-4, 
\begin{align*}
&\pr\{\wh{\pi}_t(\bm{X}_t) \neq \pi_t^*(\bm{X}_t)\} 
\leq c_1 e^{-c_2 n}, \\
&\pr\{V_t(\pi_t^*) - V_t(\wh{\pi}_t) \geq 
c_3 n^{-1/2} \tau\}
\leq e^{-\tau},
\end{align*}
where $c_i$'s are constants independent of $n$ and $\tau$.
\end{theorem}

In both theorems, the convergence rates are independent of
$\zeta_n$ and $\eta_n$.
However, the choice of $\zeta_n$ and $\eta_n$ has an impact on the 
limiting treatment regime $\pi_t^*$.
In practice, we suggest to tune $\bm{\gamma}_n$ and $\lambda_n$
by minimizing the cross validated mean squared error in the 
kernel ridge regression,
and tune $\zeta_n$ and $\eta_n$ by maximizing the 
cross validated value of the regime.

\section{Simulation Studies}
\label{sec:simulation}
We conducted a series of simulation experiments 
to examine the empirical performance of the
proposed method.  Five scenarios were considered.  The first four came
from \citet{zhao2015new} and the fifth was adapted from
\citet{murphy2003optimal}.  Scenario \RomanNum{1} consists of two
stages, two treatment options at each stage, and the covariates
exhibit nonlinear effects: 
$\bm{S}_1 = (S_{1,1}, \dots, S_{1,50})^\T$, are independent standard
normal
random variables; $A_1$ is 
$\mathrm{Uniform}\{-1, 1\}$; $Y_1$ is 
$\mathrm{Normal}(\mu_1, 1)$, where $\mu_1 = 0.5 S_{1,3} A_1$;
$S_2$ is empty; $A_2$
is  $\mathrm{Uniform}\{-1, 1\}$; $Y_2$ is $\mathrm{Normal}(\mu_2, 1)$ where
$\mu_2 = \{ (S_{1,1}^{2} + S_{1,2}^{2} - 0.2) (0.5 - S_{1,1}^{2} -
S_{1,2}^{2}) + Y_1\} A_2$.
Scenario \RomanNum{2} consists of time-varying covariates.  In this
scenario: $\bm{S}_1$, $A_1$ and $A_2$ were generated in the same way
as scenario \RomanNum{1}; $Y_1$ is $\mathrm{Normal}(\mu_1, 1)$, where
$\mu_1 = (1 + 1.5 S_{1,3}) A_1$, $\bm{S}_2 = (S_{2,1}, S_{2,2})^\T$;
$S_{2,1}$ is Bernoulli with success probability
 $1-\Phi(1.25 S_{1,1} A_1)$;  $S_{2,2}$ is Bernoulli with success
 probability $1-\Phi(-1.75 S_{1,2} A_1)$; and $Y_2$ is
 $\mathrm{Normal}(\mu_2, 1)$, where
 $\mu_2 = (0.5 + Y_1 + 0.5 A_1 + 0.5 S_{2,1} - 0.5 S_{2,2}) A_2$. 
 In Scenario \RomanNum{3}, $A_1$, $A_2$, $A_3$
 are $\mathrm{Uniform}\{-1, 1\}^3$;
 $S_{1,1}, S_{1,2}, S_{1,3}$  are $i.i.d.$ $\mathrm{Normal}(45, 15^2)$;
 ${S}_2$
  is $\mathrm{Normal}(1.5 S_{1,1}, 10^2)$;  ${S}_3$ is
 $\mathrm{Normal}(0.5 S_2, 10^2)$;  $Y_1 = Y_2 = 0$, and $Y_3$ is
 $\mathrm{Normal}(\mu_3, 1)$, where
 $\mu_3 = 20 - |0.6 S_{1,1} - 40| \{I(A_1>0) - I(S_{1,1} > 30)\}^2 -
 |0.8 S_2 - 60| \{I(A_2 > 0) - I(S_2 > 40)\}^2 - |1.4S_3 - 40| \{
 I(A_3>0) - I(S_3 > 40)\}^2$.
Scenario \RomanNum{4} is the same as Scenario \RomanNum{3} except that
many noise variables were added.  In addition to $S_{1,1}$,
$S_{1,2}$ and $S_{1,3}$, we generated $S_{1,4}, \dots, S_{1,50}$
$i.i.d.$ from 
$\mathrm{Normal}(45, 15^2)$.  
Scenario \RomanNum{5} involves ten stages and multiple
treatment options at each stage.  See \citet{murphy2003optimal} for
background and motivation for this scenario.  For $t \in \{1, \dots, 10\}$,
treatments were coded as a pair of values $A_t = (A_{t1}, A_{t2})^\T$, generated
as follows.  First,
$A_{t1}$ is drawn from $\mathrm{Uniform}\{0, 1\}$.  Second, if $A_{t1}
= 0$ then $A_{t2}$ is
drawn from $\mathrm{Uniform}\{0, 1, 2, 3\}$, and otherwise $A_{t2}$ is drawn
uniformly from $\{1, 2, 3\}$.  Thus, there are
$m_t = 7$ treatment candidates at each decision point.  In addition,
$U_1, \dots, U_{10}$ are $i.i.d.$ $\mathrm{Normal}(0, 0.01)$;
$S_1 = 0.5 + U_1$;
$S_t = 0.5 + 0.2 S_{t-1} - 0.07 A_{t-1, 1} A_{t-1, 2} - 0.01 (1 -
A_{t-1, 1}) A_{t-1, 2} + U_t$
for $t \geq 2$; $Y_t$  is $\mathrm{Normal}(\mu_t, 0.64)$, where
$\mu_t = 30 I(t=1) - 5 U_t - 6\{A_{t1} - I(S_t > 5/9)\}^2 - 1.5 A_{t1}
(A_{t2} - 2S_t)^2 - 1.5 (1 - A_{t1}) (A_{t2} - 5.5 S_t)^2 $
for each $t$.

In each scenario, we considered sample sizes $n=100, 200,$ and $400$.  We
generated 1000 data sets for each sample size and estimated the
optimal treatment regime using the proposed method.  In each stage, we tuned the
scaling vector in the Gaussian kernel, $\bm{\gamma}_t$, as well as the
amount of penalty, $\lambda_t$, via leave-one-out cross validation.
The cross validated error was minimized via a Quasi-Newton type
algorithm \citep{kim2010tackling} with a random starting value.
During the construction of decision lists, at each decision point we
tuned $\zeta$ and $\eta$ via five-fold cross validation over a
pre-specified grid.  We picked the combination that led to the largest
cross validated outcome.  

To form a basis for comparison, we also
implemented $Q$-learning with linear models, non-parametric
$Q$-learning with random forests, backward outcome weighted learning
(BOWL), and simultaneous outcome weighted learning
\citep[SOWL;][]{zhao2015new}.  In $Q$-learning with linear
$Q$-functions, we
fit the working 
$Q_t(\bm{X}_t, A_t) = \sum_{a \in \mathcal{A}_t} I(A_t = a)
\bm{X}_t^\T \bm{\beta}_{t, a}$.
Motivated by \citet{qian2011performance}, we imposed an $\ell_1$
penalty to reduce overfitting.  The $Q$-functions were estimated by
$\ell_1$ regularized least squares, implemented in the \texttt{R}
package \texttt{glmnet} \citep{friedman2010regularization}.  The
covariates were standardized to have mean zero and variance one before
entering the model, and the tuning parameter was selected by five-fold
cross validation.  
Our
implementation of non-parametric $Q$-learning used the \texttt{R} package
\texttt{randomForest} with default parameters settings
\citep{liaw2002classification}.  We implemented BOWL and SOWL
according to the descriptions in \citet{zhao2015new}.  Linear kernels
were used, and the amount of regularization was chosen by five-fold
cross validation.  Note that BOWL and SOWL assume binary treatment
options and thus are not applicable in Scenario \RomanNum{5}.

We measure the quality of an estimated treatment regime by the mean outcome under
that treatment regime; we approximate this mean outcome using an independent test set of size $10^5$.  
The results are displayed in Table~1. 
In Scenario
\RomanNum{1}, the second stage $Q$-function is highly nonlinear,
and most
methods tended to assign a single treatment to all patients in the second stage,
leading to a mean outcome of $6.70$.  In contrast, the proposed method
is able to correctly individualize treatment as the sample size
increased and thus produce a higher mean outcome.  In Scenario
\RomanNum{2}, both $Q$-functions at the first and the second stages
are linear.  Hence, as expected, $Q$-learning with linear models
performs best. Nevertheless, the proposed method and non-parametric $Q$-learning
perform well and shows marked improvement over BOWL and
SOWL.  In Scenario \RomanNum{3}, $Q$-learning with linear models
suffers from model misspecification whereas the proposed method 
and non-parametric $Q$-learning both perform well. 
Furthermore, although the $Q$-functions are
complicated, the optimal treatment regime consists of linear functions of
covariates.  Hence, both BOWL and SOWL perform well in this scenario. 
Recall that scenario \RomanNum{4} is the same as scenario\RomanNum{3}
except for the addition of many noise variables.  Thus, the results 
for scenario \RomanNum{4} demosntrate a sensitivity to noise variables
in BOWL and SOWL.  One possible reason for this is that both
BOWL and SOWL utilizes $\ell_2$ penalties, which fails to exclude
noise variables.  In Scenario \RomanNum{5}, the proposed method
outperforms competing methods, especially when the sample size is
small.  The reason might be due to the nonparametric
estimation of $Q$-functions and the simple form of decision list
compared to a random forest, as simpler treatment regimes tends to have better
generalizability.


\begin{table}
\caption{Simulation results. Given a scenario and a sample size,
each method constructed $1000$ treatment regimes, one per each simulated dataset.
The number in each cell is the outcome under the estimated treatment regime,
averaged over $1000$ replications, with standard deviation in parentheses. In the header, $n$ is the sample size, 
DL refers to the proposed decision list based approach, 
$Q$-lasso refers to the $Q$-learning approach with linear model and lasso penalty, 
$Q$-RF refers to the $Q$-learning approach using random forest
\label{tab:result}}
\begin{center}
\begin{tabular}{ccrrrrr}
Scenario & $n$ & \cellcenter{DL} & \cellcenter{$Q$-lasso} & \cellcenter{$Q$-RF} & \cellcenter{BOWL} & \cellcenter{SOWL} \\ 
  \hline
\RomanNum{1} & 100 & 6.63 (0.24) & 6.55 (0.58) & 6.70 (0.05) & 6.70 (0.05) & 6.70 (0.05) \\ 
\RomanNum{1} & 200 & 6.73 (0.24) & 6.64 (0.33) & 6.70 (0.05) & 6.70 (0.05) & 6.70 (0.05) \\ 
\RomanNum{1} & 400 & 6.94 (0.16) & 6.66 (0.26) & 6.70 (0.05) & 6.70 (0.05) & 6.70 (0.05) \\ 
\RomanNum{2} & 100 & 3.66 (0.10) & 3.68 (0.08) & 3.41 (0.17) & 3.15 (0.05) & 2.77 (0.52) \\ 
\RomanNum{2} & 200 & 3.71 (0.04) & 3.73 (0.04) & 3.62 (0.12) & 3.22 (0.08) & 2.84 (0.33) \\ 
\RomanNum{2} & 400 & 3.73 (0.03) & 3.75 (0.02) & 3.71 (0.04) & 3.37 (0.14) & 2.91 (0.28) \\ 
\RomanNum{3} & 100 & 14.49 (2.77) & 5.42 (4.54) & 12.94 (2.07) & 10.65 (2.40) & 10.27 (2.33) \\ 
\RomanNum{3} & 200 & 17.42 (1.42) & 7.88 (1.63) & 15.79 (1.59) & 13.09 (2.20) & 12.98 (1.88) \\ 
\RomanNum{3} & 400 & 18.60 (0.71) & 8.41 (0.65) & 18.02 (0.73) & 15.33 (1.56) & 16.22 (1.58) \\ 
\RomanNum{4} & 100 & 13.38 (3.14) & 4.54 (5.17) & 11.47 (2.31) & 6.72 (1.71) & 6.04 (2.18) \\ 
\RomanNum{4} & 200 & 17.33 (1.87) & 7.69 (2.33) & 14.82 (1.75) & 8.90 (1.13) & 8.34 (1.99) \\ 
\RomanNum{4} & 400 & 18.84 (0.70) & 8.61 (0.96) & 17.04 (1.02) & 10.75 (0.68) & 9.38 (2.30) \\ 
\RomanNum{5} & 100 & 23.68 (1.09) & 12.97 (3.40) & 17.83 (1.63) & $-$ & $-$ \\ 
\RomanNum{5} & 200 & 25.94 (0.51) & 13.80 (2.57) & 21.60 (1.28) & $-$ & $-$ \\ 
\RomanNum{5} & 400 & 26.80 (0.29) & 16.65 (1.71) & 24.73 (0.65) & $-$ & $-$ \\ 
\end{tabular}
\end{center}
\end{table}

\section{Data Analysis}
\label{sec:data-analysis}

As an illustration of the proposed method,
we use data from the Systematic Treatment
Enhancement Program for Bipolar Disorder (STEP-BD)
to estimate an interpretable treatment regime for treating bipolar disorder
\citep[][]{sachs2003rationale}.
We focus on the randomized acute depression (RAD) pathway in STEP-BD,
which is a Sequential Multiple Assignment Randomized Trial (SMART) and
provides the data needed to build treatment regimes.  One purpose of STEP-BD is to
assess the effectiveness of adding antidepressants to mood stabilizers
in treating patients with bipolar disorder.  Although antidepressants
were often assigned to supplement mood stabilizers in practice, it was
found that the adjunctive antidepressant medication did not show much
improvement over the use of mood stabilizers alone
\citep{sachs2007effectiveness}.  Thus, it is of scientific interest to
tailor the use of antidepressants based on individual and
time-dependent characteristics.

The RAD pathway in STEP-BD is a randomized trials with two stages.  At
both stages, patients always received one or more mood stabilizers
chosen by their psychiatrists.  In addition, they might receive one
antidepressant in the form of bupropion or paroxetine.  At week 0,
patients were randomized to receive bupropion, paroxetine or placebo
with probability 0.25, 0.25 and 0.5, respectively.  After 6 weeks,
patients returned to their psychiatrists for evaluation on response
status.  In another 6 weeks, responders continued their initial
treatments, non-responders who received either bupropion or paroxetine
initially were offered an increased dose, and non-responders who
received placebo initially were randomized to received bupropion or
paroxetine with equal probability.  At week 12, patients returned to
their psychiatrists for final measurements.

In this clinical trial, 
the covariate $\bm{X}_1$ of a patient consists of his/her 
age, gender, marital status, education level, employment status, 
bipolar type, nature of the episode prior to the current depressive episode, 
summary score for depression (SUM-D) at baseline,
and summary score for mood elevation (SUM-ME) at baseline.
The treatment $A_1$ takes three values: bupropion, paroxetine and placebo.
The covariate $\bm{X}_2$ consists of SUM-D at week 6, SUM-ME at week 6, and indicators for nine different adverse events at week 6.
The treatment $A_2$ is either bupropion or paroxetine for non-responders who received placebo in the first stage. 
For other patients, $A_2$ is the same as $A_1$.
The outcomes are $Y_1 = 0$ and $Y_2 = \text{SUM-D at week 12}$.
Note that smaller values of SUM-D and SUM-ME indicates better clinical status.
A complete description of these variables is provided in the
 Supplemental Materials.

We apply the propose method to estimate an interpretable treatment regime.
For simplicity, we only include patients with complete baseline and stage~1 information. 
And we use the last-value-carry-forward strategy if the SUM-D at week 12
is missing.
The estimated optimal decision rule at the first stage is:
\begin{align*}
&\texttt{If } \text{SUM-D at week 0} > 8.625 \texttt{ then } \text{bupropion};  \\
&\texttt{else if } \text{SUM-D at week 0} \leq 4.875 \texttt{ and }
\text{race is not white} \texttt{ then } \text{paroxetine};  \\
&\texttt{else } \text{placebo}. 
\end{align*}
The estimated regime suggests that the baseline SUM-D is informative
in treatment selection.  Recall that smaller values of SUM-D indicate
lower symptoms. Hence, an interpretation of the estimated
regime is: patients with severe depression symptoms should receive
bupropion, while non-white patients with minor depression symptoms
should receive paroxetine.  Although applying an antidepressant
medication to all patients did not lead to a better mean outcome
relative to
not
applying antidepressants to any of the patients
\citep{sachs2007effectiveness}, the estimated regime indicates that
personalizing the use of antidepressants based on SUM-D may improve
the overall mean outcome.
The estimated optimal decision rule at the second stage is:
\begin{align*}
&\texttt{If } \text{SUM-ME at week 6} \leq 0.875 \texttt{ then } \text{bupropion};  \\
&\texttt{else if } \text{SUM-D at week 6} > 8.5\texttt{ then } \text{bupropion};  \\
&\texttt{else } \text{paroxetine}. 
\end{align*}
From this rule it can be seen that patients with large SUM-D or low
SUM-ME are assigned to buproprion. 


\section{Discussion}
\label{sec:discussion}
The current trend in methodological research for estimation of optimal
treatment regimes seems to be the development of increasingly flexible
models to mitigate risk of model misspecification.  This trend is
aligned with the notion that an estimated optimal regime will be used
to make treatment decisions for future patients.  However, in many
settings an estimated optimal regime is not used to make treatment
decisions but rather is used to generate hypotheses and inform future
research.  Indeed, our view is that the development of a precision 
medicine strategy should be the culmination of an iterative process 
of hypothesis generation and validation.  With this perspective, 
the ability to interpret and estimated optimal regime in a domain
context is paramount.  

We used list-based regimes to ensure interpretability of the
estimated regimes.  Our proposed estimation algorithm combines
non-parametric $Q$-learning with policy-search and consistently 
estimates the optimal regime under mild assumptions.  In principle,
the proposed estimation framework could be used to estimate
interpretable optimal regimes of other forms, e.g., more general
tree structures or rule-based systems.   Nevertheless,  the simplicity 
of list-based regimes that ensures parsimony and interpretability also
appears to have regularizing effect that improves generalization
performance.

The recognition that estimated optimal regimes are often not used
directly to select treatments for patients but instead are part of an
iterative, collaborative process opens many new lines of research
beyond estimation of interpretable regimes.  These include
methods for visualization, models for shared-decision making,
models for patient preference and utility construction, and methods
for constructing prediction sets for outcome trajectories in 
multistage decision problems.
We are currently pursuing several of these research areas.

\bibliographystyle{Chicago}
\bibliography{references}

\clearpage
\appendix

\begin{center}
\bf\Large
Supplementary Materials to ``Interpretable Dynamic Treatment Regimes''
\end{center}

\section{Proofs}

\subsection{Notation}

For vectors $\bm{u}, \bm{v} \in \R^q$, define component-wise
operations $\bm{u}^{p} = (u_1^{p}, \dots, u_q^{p})^T$, $p \in \R$, and
$\bm{u} \circ \bm{v} = (u_1 v_1, \dots, u_q v_q)^T$.  For
$V \subset \R^q$, define $u \circ V = \{u \circ v: v \in V\}$.  In
addition, $\bm{u}$ is said to be positive if its every component is
positive.

Let $\bm{O}_i$ be the collection of random variables associated with the $i$th subject. For any function $f$, define $\Pn f = n^{-1} \sum_{i=1}^n f(\bm{O}_i)$.
For any measurable function $f$ defined on $D \subset \R^q$,
we write $\lVert f \rVert_2 = \big( \int_D f^2 d \mu \big)^{1/2}$
and $\lVert f \rVert_\infty = \inf\{ t \in \R: \mu(|f|>t)=0 \}$,
where $\mu$ is the Lebesgue measure on $D$.
Let $(T, d)$ be a metric space and $S$ be a subset of $T$. 
For $\varepsilon > 0$, the $\varepsilon$-covering number of $S$ is defined by 
$\mathcal{N}(S, d, \varepsilon) = \inf \{n \geq 1: \text{ there exists } t_1, \cdots, t_n \in T \text{ such that } S \subset \bigcup^{n}_{i=1} B(t_i, \varepsilon)\}$, 
where $\inf \emptyset = \infty$ and $B(t, \varepsilon) = \{u \in T: d(u, t) \leq \varepsilon \} $ is the a ball with center $t$ and radius $\varepsilon$.
If $(T, \lVert \cdot \rVert)$ is a normed vector space, 
the $\varepsilon$-covering number is defined by
viewing $T$ as a metric space with induced metric
$d(s, t) = \lVert s - t \rVert$. 
Let $(T, \lVert \cdot \rVert)$ be a normed vector space.
The unit ball of $T$ is defined by
$\mathcal{B}_T = \{t: \lVert t \rVert \leq 1\}$.
Given a scalar $w \in \R$ and a set $S \subset T$,
define $w S = \{w s: s \in S \}$.

In the following proofs, $c$ and $c_i$ denote generic constants.

\subsection{Concentration inequalities}

We first state Talagrand's inequality (\citealp[][Theorem 2.3]{bousquet2002bennett}; see also \citealp[][Theorem 3]{massart2000constants} and \citealp[][Theorem 12.5]{boucheron2013concentration}).

\begin{proposition}
\label{thm:talagrand0}
Let $\mathcal{F}$ be a countable set of functions. 
Suppose $\E f =0$, $\E f^2 \leq V$, $\lVert f \rVert_{\infty} \leq B$ for all $f \in \mathcal{F}$. 
Denote $Z = \sup_{f \in \mathcal{F}}|\Pn f|$. 
Then for all $\tau > 0$, 
\[
\Pr \left[ Z \geq \E Z + \left\{ \frac{2 V \tau + 4B \tau(\E Z)}{n} \right\}^{1/2} + \frac{B\tau}{3n} \right] \leq e^{-\tau}.
\]
\end{proposition}

\begin{corollary}
\label{thm:talagrand}
Under the conditions in Proposition~\ref{thm:talagrand0},
\[
\Pr \left\{ Z \geq 2 \E Z + \left( \frac{2V\tau}{n} \right)^{1/2} 
+ \frac{2B\tau}{n} \right\} \leq e^{-\tau}.
\]
\end{corollary}

\begin{proof}
It is clear that
\[
\left\{ \frac{2 V \tau + 4B(\E Z) \tau}{n} \right\}^{1/2}
\leq \left(\frac{2 V \tau}{n}\right)^{1/2} + 
\left\{ \frac{4B\tau (\E Z)}{n} \right\}^{1/2}
\leq \left(\frac{2 V \tau}{n}\right)^{1/2} +
\frac{B\tau}{n} + \E Z.
\]
Note that we use a larger constant for simplicity.
\end{proof}

When the variance of $f$ is unavailable, we have the following proposition \citep[][Theorem 12.1]{boucheron2013concentration}.

\begin{proposition}
\label{thm:hoeffding}
Let $\mathcal{F}$ be a countable set of functions. 
Suppose $\E f =0$, $\lVert f \rVert_{\infty} \leq B$ for all $f \in \mathcal{F}$. 
Denote $Z = \sup_{f \in \mathcal{F}}|\Pn f|$. 
Then for all $\tau > 0$, we have 
\[
\Pr \left\{ Z \geq \E Z + \left( \frac{2 B^2 \tau}{n} \right)^{1/2}\right\} \leq e^{-\tau}.
\]
\end{proposition}

Next, we establish bounds on $\E Z$.

\begin{proposition}
\label{thm:sup-tala}
Let $\mathcal{F}$ be a countable set of functions which contains the zero function. 
Assume 
\[\sup_{Q} \log \mathcal{N}(\mathcal{F}, \lVert \cdot \rVert_{L^2(Q)}, \varepsilon) \leq \psi(\varepsilon)
\]
for some function $\psi(\cdot)$,
where the supremum is taken over all discrete probability measures $Q$. 
Suppose $\E f = 0$, $\E f^2 \leq V$, $\lVert f \rVert_{\infty} \leq B$ for all $f \in \mathcal{F}$. Denote $Z = \sup_{f \in \mathcal{F}} |\Pn f|$. Then we have 
\[
\E Z \leq 1024 \left(\frac{B J_V}{n}\right) + 64 \left(\frac{V J_V}{n}\right)^{1/2},
\] 
where $J_V = \int^{1}_{0} \psi(V^{1/2}\varepsilon)\,d\varepsilon$. 
\end{proposition}

\begin{proof}

Without loss of generality, we assume $B=1$. The general case can be obtained by scaling $f$. The proof extends the idea in \citet[][Lemma 13.5]{boucheron2013concentration}.

Let $\sigma_1, \dots, \sigma_n$ be i.i.d.\@ Rademacher random variables, i.e., $\Pr(\sigma = 1) = \Pr(\sigma=-1) = 1/2$. 
By the symmetrization inequality \citep[][Lemma 2.3.1]{van1996weak}, 
we have 
$\E(n^{1/2}Z) \leq 2\E \sup_f |n^{1/2} \Pn \sigma f|$. 

Conditional on all random variables except $\sigma_i$s, by Hoeffding's inequality, 
the process $n^{1/2}\Pn \sigma f$ is subgaussian with respect to the metric $\lVert f - g\rVert_{L^2(\Pn)} = \{\Pn(f-g)^2\}^{1/2}$ . Hence the chaining technique \citep[][Corollary 2.2.8]{van1996weak} implies 
\[
\E_{\sigma} \sup_{f} |n^{1/2} \Pn \sigma f| \leq 4 \int_{0}^{\eta_n} \left\{ \log \mathcal{N}(\mathcal{F}, \lVert \cdot \rVert_{L^2(\Pn)}, \varepsilon ) \right\}^{1/2} \, d\varepsilon,
\]
where $\E_{\sigma}$ denote the expectation with respect to $\sigma_1, \dots, \sigma_n$ only and $\eta^{2}_{n} = \max\{ \sup_f (\Pn f^2), V\}$. Hence, we obtain 
\[
\E_{\sigma} \sup_f |n^{1/2} \Pn \sigma f| \leq 4 \int^{\eta_n}_{0} \psi^{1/2}(\varepsilon)\,d\varepsilon = 4\eta_n \int_{0}^{1}\psi^{1/2} (\eta_n \varepsilon) \, d\varepsilon \leq 4 \eta_n \int^{1}_{0} \psi^{1/2}(V^{1/2}\varepsilon)\,d\varepsilon.
\]
Because $\log\mathcal{N}(\mathcal{F}, \lVert \cdot \rVert_{L^2(\Pn)}, \varepsilon) \leq \psi(\varepsilon)$ and $\psi(\varepsilon)$ is a decreasing function in $\varepsilon$.

Taking the other layer of expectation, we get 
\[
\E(n^{1/2} Z) \leq 8(\E \eta_n) \int^{1}_{0} \psi^{1/2}(V^{1/2}\varepsilon)\,d\varepsilon \leq 8 (\E \eta_{n}^{2})^{1/2} J_V^{1/2}
\]
by Jensen's inequality. Also, we have 
$
\E \eta^2_n \leq \E \sup_{f} |\Pn f^2 - \E f^2| + V
$ since $\E f^2 \leq V$ for all $f$.
By the symmetrization inequality \citep[][Lemma 2.3.1]{van1996weak}, 
we have
$
\E \sup_f |\Pn f^2 - \E f^2| \leq 2 \E \sup_f |\Pn \sigma f^2|
$.
By the contraction inequality \citep[][Proposition A.3.2]{van1996weak}
and $\lVert f \rVert_\infty \leq 1$, we have
$
\E \sup_{f} |\Pn \sigma f^2| \leq 4 \E \sup_f |\Pn \sigma f|
$.
By the desymmetrization inequality \citep[][Lemma 2.3.6]{van1996weak}, 
we have
$
\E \sup_f |\Pn \sigma f| \leq 2 \E \sup_f |\Pn f|
$.
Combining these inequalities, yields 
$
\E \eta^{2}_{n} \leq 16\E Z + V
$.

Therefore, 
\[
n^{1/2} \E Z \leq 8 (16\E Z + V)^{1/2} J_V^{1/2}.
\]
Solving for $\E Z$, shows
$
\E Z \leq (2n)^{-1} \{ a + (a^2+4nb)^{1/2}\} \leq n^{-1}a + n^{-1/2}b^{1/2}
$
with $a = 1024 J_V$ and $b=64V J_V$. Hence,
$
\E Z \leq 1024n^{-1}J_V + 64 n^{-1/2}V^{1/2} J_V^{1/2}
$.
\end{proof}

\begin{proposition}
\label{thm:sup-hoef}
Let $\mathcal{F}$ be a countable set of functions which contains the zero function. 
Assume 
\[\sup_{Q} \log \mathcal{N}(\mathcal{F}, \lVert \cdot \rVert_{L^2(Q)}, \varepsilon) \leq \psi(\varepsilon)
\]
for some function $\psi(\cdot)$,
where the supremum is taken over all discrete probability measures $Q$. 
Suppose $\E f = 0$, $\lVert f \rVert_{\infty} \leq B$ for all $f \in \mathcal{F}$. Denote $Z = \sup_{f \in \mathcal{F}} |\Pn f|$. Then we have 
\[
\E Z \leq 8 \left(\frac{B^2 J_B}{n}\right)^{1/2},
\] 
where $J_B = \int^{1}_{0} \psi(B \varepsilon)\,d\varepsilon$. 
\end{proposition}

\begin{proof}
Just apply the trivial bound $|\eta_n| \leq B$ in the proof of Proposition~\ref{thm:sup-tala}.
\end{proof}

Though all the propositions in this subsection assume that $\mathcal{F}$ is countable, 
they all apply if $\mathcal{F}$ is uncountable and separable 
as $\pr\left( \sup_{f \in \mathcal{F}} |\Pn f| = \sup_{f \in \mathcal{F}'} |\Pn f| \right) = 1$ 
for some countable subset $\mathcal{F}' \subset \mathcal{F}$.

\subsection{Properties of the RKHS}

We establish several useful properties of the RKHS $\H$
induced by the Gaussian kernel with 
individual scaling factors for each dimension
\[
K_{\bm{\gamma}}(\bm{x}, \bm{z}) = \exp\left\{
-\sum_{j=1}^{d} \gamma_{j} (x_j - z_j)^2 \right\},
\]
where $\bm{x}, \bm{z} \in D \subset \R^q$.
The lemmas below extend 
the properties of Gaussian kernel with a single scaling factor.

We may omit $\bm{\gamma}$ and write $K(\cdot, \cdot)$ 
when the value of $\bm{\gamma}$ is clear from the context.
Similarly, to emphasize the dependence of $\H$ 
on the parameter $\bm{\gamma}$ and the domain $D$, 
we may write 
$\Hgamma$, $\H(D)$, or $\Hgamma(D)$.

The following lemma provides a feature map of the Gaussian kernel.

\begin{lemma}
\label{thm:rkhs-feature}
Define the function $\phi_{\bm{\gamma}}^{\bm{x}}: \R^q \to L^2(\R^q)$ by
\[
\phi_{\bm{\gamma}}^{\bm{x}}(\bm{u}) = 
\left(\frac{4}{\pi}\right)^{q/4} \left(\prod_{j=1}^{q} \gamma_j\right)^{1/4} 
\exp\left\{- \sum_{j=1}^{q} 2\gamma_j(x_j - u_j)^2 \right\}, 
\; \bm{x} \in D, \; \bm{u} \in \R^q.
\]
Then $\phi_{\bm{\gamma}}^{\bm{x}}$ is a feature map of $K_{\bm{\gamma}}(\bm{x}, \bm{z})$.
\end{lemma}

\begin{proof}
Straightforward calculation similar to \citet[Lemma~4.45]{steinwart2008support} gives
$\langle \phi_{\bm{\gamma}}^{\bm{x}}, \phi_{\bm{\gamma}}^{\bm{z}} \rangle_{L^2(\R^q)} 
= K_{\bm{\gamma}}(\bm{x}, \bm{z})$.
By definition, $\phi_{\bm{\gamma}}^{\bm{x}}$ is a feature map.
\end{proof}

The following lemma shows that $\Hgamma(D)$
can be embedded into $\H_{\wt{\bm{\gamma}}}(D)$
if $\gamma_j < \wt{\gamma}_j$ for all $j = 1, \dots, q$.

\begin{lemma}
\label{thm:rkhs-inclusion}
Let $\bm{\gamma}$, $\wt{\bm{\gamma}}$ be two positive vectors satisfying
$\gamma_j < \wt{\gamma}_j$ for all $j$.
If $f \in \Hgamma$, then $f \in \H_{\wt{\bm{\gamma}}}$ and 
$
\lVert f \rVert_{\H_{\wt{\bm{\gamma}}}} \leq \left( \prod_{j=1}^q\wt{\gamma}_j \right)^{1/4} \left( \prod_j {\gamma}_{j=1}^q \right)^{-1/4} \lVert f \rVert_{\Hgamma}.
$
\end{lemma}

\begin{proof}
We follow the strategy in \citet[Theorem 4.46]{steinwart2008support}.
Since $f \in \Hgamma$, 
by \citet[Theorem 4.21]{steinwart2008support}, 
there exists $g \in L^2\left( \R^q \right)$ such that 
$f(\bm{x}) = \langle\phi^{\bm{x}}_{\bm{\gamma}}, g \rangle_{L^2(\R^q)}$ for all $\bm{x} \in D$. 

Given $\bm{s} \in \R^q$ with $s_j > 0$ for all $j$, 
define the operator $W_s: L^2\left( \R^q \right) \to L^2\left( \R^q \right)$ by
\[
(W_{\bm{s}} g)(\bm{v}) = 
\int_{\R^q} \pi^{-q/2} \bigg( \prod^{q}_{j=1} s_j \bigg)^{-1/2} 
\exp\bigg\{ -\sum_{j=1}^{q} s_j^{-1} (v_j - u_j)^2 \bigg\} g(\bm{u}) d \bm{u}, 
\; \text{for }\bm{v} \in \R^q.
\]
For any $g \in L^2(\R^q)$ and any $\bm{v} \in \R^q$, 
straightforward calculation using properties of normal densities
shows $(W_{\bm{s}_1} W_{\bm{s}_2} g)(\bm{v}) = (W_{\bm{s}_1 + \bm{s}_2}g)(\bm{v})$,
hence,  $W_{\bm{s}_1}W_{\bm{s}_2} = W_{\bm{s}_1+\bm{s}_2}$. 

Define $\bm{\tau} = (\tau_1, \dots, \tau_q)^\T$ and 
$\wt{\bm{\tau}} = (\wt{\tau}_1, \dots, \wt{\tau}_q)^\T$, 
where $\tau_j = 1/\gamma_j$ and $\wt{\tau}_j = 1/ \wt{\gamma}_j$. 
The assumption $\gamma_j < \wt{\gamma}_j$ implies $\tau_j > \wt{\tau}_j$. 
We observe that 
\[
f = \langle \phi_{\bm{\gamma}}^{\bm{x}}, g \rangle_{L^2(\R^q)}
= (W_{\bm{\tau}/2} g)  \cdot \pi^{q/4} \left( \prod^{q}_{j=1} \gamma_j \right)^{-1/4}.
\]
Because  
\[
W_{\bm{\tau}/2} g = W_{\wt{\bm{\tau}}/2} W_{(\bm{\tau} - \wt{\bm{\tau}})/2} g
= \langle \phi_{\wt{\bm{\gamma}}}^{\bm{x}} , W_{(\bm{\tau} - \wt{\bm{\tau}})/2 } g\rangle_{L^2(\R^q)}
\cdot \pi^{-q/4} \left( \prod^{q}_{j=1} \wt{\gamma}_j \right)^{1/4},
\]
it follows that
\[
f = \langle \phi_{\wt{\bm{\gamma}}}^{\bm{x}} , W_{(\bm{\tau} - \wt{\bm{\tau}})/2 } g\rangle_{L^2(\R^q)}
\cdot \left( \prod^{q}_{j=1} {\gamma}_j \right)^{-1/4} \left( \prod^{q}_{j=1} \wt{\gamma}_j \right )^{1/4}.
\]
By \citet[Theorem 4.21]{steinwart2008support}, $f \in \H_{\wt{\bm{\gamma}}}$. 

Moreover, 
$\lVert f \rVert_{\Hgamma} = \lVert g \rVert_{L^2(\R^q)}$ 
and 
$\lVert f \rVert_{\H_{\wt{\bm{\gamma}}}} = 
\lVert W_{(\bm{\tau} - \wt{\bm{\tau}})/2} g \rVert_{L^2(\R^q)} 
\cdot \left( \prod_{j=1}^{q} \gamma_j \right)^{-1/4} \left( \prod_{j=1}^{q} \wt{\gamma}_j \right)^{1/4} $. 
By Young's inequality, $\lVert W_{(\bm{\tau} - \wt{\bm{\tau}})/2} g \rVert_{L^2(\R^q)} \leq \lVert g \rVert_{L^2(\R^q)}$. 
Hence, 
\[
\lVert f \rVert_{\H_{\wt{\bm{\gamma}}}} 
\leq \lVert g \rVert_{L^2(\R^q)} \left( \prod^{q}_{j=1} \gamma_j \right)^{-1/4} \left( \prod^{q}_{j=1} \wt{\gamma}_j \right)^{1/4} \\
\leq \lVert f \rVert_{\Hgamma} \left( \prod^{q}_{j=1} \gamma_j \right)^{-1/4} \left( \prod^{q}_{j=1} \wt{\gamma}_j \right)^{1/4}.
\]
\end{proof}

The following lemma establishes an isometric isomorphism
between $\H_{\bm{\alpha}^{-2} \circ \bm{\gamma}} (\bm{\alpha} \circ D)$ and $\Hgamma(D)$ for any fixed $\bm{\alpha}$.

\begin{lemma}
\label{thm:rkhs-scaling}
Let $\bm{\alpha}$ be an arbitrary positive vector.
We define a mapping $\tau_{\bm{\alpha}}: L^{\infty}(D) \to L^{\infty}(\bm{\alpha} \circ D)$ as follows:
given a function $f \in L^{\infty}(D)$,
let $\tau_{\bm{\alpha}}(f) (\bm{x}) = f(\bm{\alpha}^{-1} \circ \bm{x})$ for $\bm{x} \in \bm{\alpha} \circ D$.
Then, for all $f \in \Hgamma(D)$, we have
$\tau_{\bm{\alpha}}(f) \in \H_{\bm{\alpha^{-2} \circ \bm{\gamma}}} (\bm{\alpha} \circ D)$
and
$\lVert \tau_{\bm{\alpha}}(f) \rVert_{\H_{\bm{\alpha^{-2} \circ \bm{\gamma}}}(\bm{\alpha} \circ D)} = \lVert f \rVert_{\H_{\bm{\gamma}(D)}}$. 
\end{lemma}

\begin{proof}
  It is easy to verify that the arguments in \citet[Proposition
  4.37]{steinwart2008support} remain valid when scalar multiplication
  is replaced by component-wise multiplication between vectors.
\end{proof}

The following lemma computes the covering number of the unit ball in $\Hgamma(D)$.

\begin{lemma} 
\label{thm:rkhs-covering}
Suppose $D \subset s \mathcal{B}_{\R^q}$.
For any integer $m \geq 1$, 
\[
\log \mathcal{N}\{\mathcal{B}_{\Hgamma(D)}, \lVert \cdot \rVert_{\infty}, \varepsilon\}
\leq c_{m, q, s} \prod^{q}_{j=1} (1+\gamma_j)^{1/2} \varepsilon^{-q/m},
\] 
where $c_{m, q, s}$ is a constant that depends on $m$, $q$ and $s$ only.
\end{lemma}

\begin{proof}
  Let $\bm{1}$ be the vector of ones. By Lemma~3, $\Hgamma(D)$ is
  isometric isomorphic to $\H_{\bm{1}}(\bm{\gamma}^{1/2} \circ D)$.
  Thus, it suffices to compute the covering number for
  $\H_{\bm{1}}(\bm{\gamma}^{1/2} \circ D)$.

Define $\wt{D} = \bm{\gamma}^{1/2} \circ D$.
It is shown that
$\H_{\bm{1}}(\wt{D})$ can be embedded into $\mathbb{C}^{m}(\wt{D})$
\citep[Theorem 6.26]{steinwart2008support}.
By \citet[Corollary 4.36]{steinwart2008support}, 
the embedding map from $\H_{\bm{1}}(\wt{D})$ to $\mathbb{C}^{m}(\wt{D})$ is continuous, and hence bounded. 
Thus, there exists a constant $c_1$ which depends only on $m$ such that 
$\lVert f \rVert_{\mathbb{C}^m(\wt{D})} 
\leq c_1 \lVert f \rVert_{\H_{\bm{1}}(\wt{D})}$ for all $f \in \H_{\bm{1}}(\wt{D})$. 
Hence, we have
\[
\mathcal{N}\{\mathcal{B}_{\bm{H}_{\bm{1}}(\wt{D})}, \lVert \cdot \rVert_{\infty}, \varepsilon\}
\leq \mathcal{N}(c_1 \mathcal{B}_{\mathbb{C}^m(\wt{D})}, \lVert \cdot \rVert_{\infty}, \varepsilon\}
= \mathcal{N}(\mathcal{B}_{\mathbb{C}^{m}(\wt{D})}, \lVert \cdot \rVert_{\infty}, \varepsilon / c_1).
\] 

By Theorem 2.7.1 in  \citet[]{van1996weak}, there exists a constant $c_2$
that depends only on $m$ and $q$ such that 
\[
\log \mathcal{N}\{\mathcal{B}_{\mathbb{C}^m(\wt{D})}, \lVert \cdot \rVert_{\infty}, \varepsilon) 
\leq c_2 \mu(\{\bm{x}: \lVert \bm{x} - \wt{D} \rVert \leq 1 \}) \varepsilon^{-q/m},
\]
where $\mu$ is the Lebesgue measure on $\R^q$.
Because $D \subset s \mathcal{B}_{\R^q}$
and $(1 + s u^{1/2}) \leq (1 + s) (1 + u)^{1/2}$ for all $u \geq 0$,
\[
\mu(\{\bm{x}: \lVert \bm{x} - \wt{D} \rVert \leq 1 \})
\leq \prod_{j=1}^q (1 + s \lambda_j^{1/2})
\leq (1 + s)^q \prod_{j=1}^q (1 + \lambda_j)^{1/2}.
\]
\end{proof}

\subsection{Approximation error in kernel ridge regression}

Define $\wt{Y}_T = Y_T$ and $\wt{Y}_t = Y_t + Q_{t+1} \{ \bm{X}_{t+1}, \pi^{\ast}_{t+1} (\bm{X}_{t+1})\}$ for $t<T$.
Then,  $Q_t(\bm{x}, a) = E(\wt{Y}_t | \bm{X}_t=\bm{x}, A_t=a)$
for all $t$.
Fix a stage $t$ and a treatment $a \in \mathcal{A}_t$.
For notational simplicity, we shall omit the subscripts $t$ and $a$
hereafter.
Given a function $f \in L^\infty(D)$, we define
\[
\L(f) = \E \left[ I(A=a) \big\{\wt{Y} - f(\bm{X}) \big\}^2 \right]
\]
and
\[
f_0 =\argmin_{f: D \to \R,\text{ measurable}} \L(f).
\]
Simple calculations show that
$f_0(\bm{x}) = E(\wt{Y} | \bm{X}=\bm{x}, A=a)$
almost surely with respect to 
the distribution of $\bm{X}$, say $P_{\bm{X}}$. 
Hence, $f_0$ is exactly $Q_t(\cdot, a)$.
In addition, 
\[
\L(f) - \L(f_0) = E \left[ I(A = a) \left\{ f(\bm{X}) - f_0(\bm{X}) \right\}^2 \right].
\]

The function $f_0$ need not belong to the RKHS $\Hgamma$.
Nevertheless, the estimator must live in $\Hgamma$.  The following
proposition shows that it is always possible to find an
$f \in \Hgamma$ such that $f$ and $f_0$ are close.  The following
proposition is a stronger version of \citet[][Theorems~2.2
and~2.3]{eberts2013optimal} which allows multiple scaling factors and
separates signal and noise variables.

\begin{proposition}
\label{thm:krr-approx}
Suppose $f_0$ satisfies the modulus of smoothness condition
$\omega_r(f_0, s) \leq c_1 s^{r}$ for some positive integer $r$
and $\lVert f_0 \rVert_\infty \leq B$ for some constant $B$.
Let $\S$ denote the indices of signal variables in $f_0$,
i.e., the value of $f(\bm{x})$ only depends on $\bm{x}_\S$.
Then, there exists some $f \in \Hgamma$ such that
\[
\lambda \lVert f \rVert^{2}_{\Hgamma} + \lVert f - f_0 \rVert_\infty^2
\leq c \left\{ \lambda\big( \max_{j \in \S} \gamma_j \big)^{|\S|/2} \big( \max_{j \in \Sc} \gamma_j \big)^{|\Sc|/2} 
+ \big( \min_{j \in \S} \gamma_j \big)^{-r} \right \}
\]
and $\lVert f \rVert_{\infty} \leq 2^r B$,
where $c$ is some constant that depends on $c_1$, $r$, $B$ and $|\S|$ only.
\end{proposition}

\begin{proof}
Define 
\[
W(\bm{x}, \bm{u}) = \sum^{r}_{i=1} \binom{r}{i} (-1)^{i-1} \left(\frac{2}{\pi}\right)^{q/2}  \left( \prod^{q}_{j=1} \gamma_j \right)^{1/2} i^{-q} \exp \left\{- \sum^{q}_{j=1} 2\gamma_j(x_j - u_j)^2 / i^2\right\},
\]
where $\bm{x}, \bm{u} \in \R^q$. 
Let $f(\bm{x}) = \int_{\R^q} W(\bm{x}, \bm{u}) f_0 (\bm{u})\, d\bm{u}$, $\bm{x} \in D$.

Then, for every $\bm{x} \in D$,
\[
f(\bm{x}) =
\sum^{r}_{i=1} \binom{r}{i} (-1)^{i-1} \left( \frac{2}{\pi} \right)^{q/2} \left( \prod^{q}_{j=1} \gamma_j \right)^{1/2} \int_{\R^q} i^{-q} \exp \left\{ -\sum^{q}_{j=1} 2 \gamma_j (x_j - u_j)^2 / i^2 \right\} f_0(\bm{u}) d \bm{u}.
\]
Apply the change of variables $h_j=(u_j-x_j)/i$ so that
\begin{align*}
f(\bm{x}) &=
\sum^{r}_{i=1} \binom{r}{i} (-1)^{i-1} \left( \frac{2}{\pi} \right)^{q/2} \left( \prod^{q}_{j=1} \gamma_j \right)^{1/2} \int_{\R^q} \exp \left\{ -\sum^{q}_{j=1} 2 \gamma_j h_j^2 \right\} f_0(\bm{x} + i \bm{h}) d \bm{h} \\
&= \int_{\R^q} \left( \frac{2}{\pi} \right)^{q/2} \left( \prod^{q}_{j=1} \gamma_j \right)^{1/2} \exp \left( -\sum^{q}_{j=1} 2 \gamma_j h_j^2 \right) \sum_{i=1}^{r} \binom{r}{i} (-1)^{i-1} f_0(\bm{x} + i\bm{h}) d \bm{h}.
\end{align*} 
Note that
\[
f_0(\bm{x}) = \int_{\R^q} \left( \frac{2}{\pi} \right)^{q/2} \left( \prod^{q}_{j=1} \gamma_j \right)^{1/2} \exp \left( -\sum^{q}_{j=1} 2 \gamma_j h_j^2 \right) f_0(\bm{x}) d \bm{h},
\]
therefore
\[
\left| f(\bm{x}) - f_0(\bm{x}) \right| 
\leq \int_{\R^q} \left( \frac{2}{\pi} \right)^{q/2} \left( \prod^{q}_{j=1} \gamma_j \right)^{1/2} \exp \left( -\sum^{q}_{j=1} 2 \gamma_j h_j^2 \right) \left| \Delta_{\bm{h}}^r(f_0, \bm{x}) \right| d \bm{h}.
\]

Because $f_0(\bm{x}) = f_0^\ast(\bm{x}_\S)$ for some function $f_0^\ast: \R^{|\S|} \to \R$, 
\[
\left| \Delta_{\bm{h}}^r(f_0, \bm{x}) \right| 
=\left| \Delta_{\bm{h}_\S}^r(f_0^\ast, \bm{x}_\S) \right|
\leq \omega_r(f_0^\ast, \lVert \bm{h}_\S \rVert_2)
= \omega_r(f_0, \lVert \bm{h}_\S \rVert_2).
\]
Thus,
\[
\left| f(\bm{x}) - f_0(\bm{x}) \right| 
\leq \int_{\R^{|\S|}} \left( \frac{2}{\pi} \right)^{|\S|/2} \left( \prod_j \gamma_{\S, j} \right)^{1/2} \exp \left( -\sum_j 2 _{\S, j} h_{\S, j}^2 \right) \omega_r(f_0, \lVert \bm{h}_\S \rVert_2) d \bm{h}_\S.
\]
Because $\omega_r(f_0, t) \leq (1 + t/s)^r \omega_r(f_0, s)$ for all
$s, t > 0$, it follows that
\begin{align*}
\omega_r(f_0, \lVert \bm{h}_\S \rVert_2)
& \leq \left\{ 1+ \big(\min_{j \in \S} \gamma_j \big)^{1/2} \lVert \bm{h}_{\S} \rVert_2\right\}^r \omega_r\left\{f_0, \big(\min_{j \in \S} \gamma_j \big)^{-1/2}\right\} \\
& \leq \left(1 + \lVert \bm{\gamma}_{\S} \circ \bm{h}_{\S} \rVert_2 \right)^r \omega_r\left\{f_0, \big(\min_{j \in \S} \gamma_j \big)^{-1/2}\right\}.
\end{align*}
Combining these inequalities, 
\begin{multline*}
\left| f(\bm{x}) - f_0(\bm{x}) \right| \leq \omega_r\left\{f_0, \big(\min_{j \in \S} \gamma_j \big)^{-1/2}\right\} \cdot \\
\int_{\R^{|\S|}} \left( \frac{2}{\pi} \right)^{|\S|/2} \left( \prod_j \gamma_{\S, j} \right)^{1/2} \exp \left\{ -\sum_j 2 \gamma_{\S, j} h_{\S, j}^2 \right\} \left(1 + \lVert \bm{\gamma}_{\S} \circ \bm{h}_{\S} \rVert_2 \right)^r d \bm{h}_\S
\end{multline*}
Using the change of variables $t_j = \gamma_{\S, j} h_{\S, j}$, we can
see that the integral above is a constant that depends only on
$|\S|$. 
Denote this integral by $c_2$,  then
\[
\left| f(\bm{x}) - f_0(\bm{x}) \right| 
\leq c_2 \omega_r\left\{f_0, \big(\min_{j \in \S} \gamma_j \big)^{-1/2}\right\}
\leq c_1 c_2 \big(\min_{j \in \S} \gamma_j \big)^{-1/2}.
\]

Note that $W(\bm{x}, \bm{u}) = \sum_{i=1}^{r} \binom{r}{i} (-1)^{i-1} \pi^{-q/4} i^{-q/2} \left( \prod^{q}_{j=1} \gamma_j\right)^{1/4} \phi^{\bm{x}}_{\bm{\gamma}/i^2}(\bm{u})$,
where $\phi$ is the feature map defined in Lemma~1.
Let $g_i(\bm{x}) = \int_{\R^q} \phi^{\bm{x}}_{\bm{\gamma}/i^2}(\bm{u}) f_0(\bm{u}) d \bm{u}$, then
$g_i \in \H_{\bm{\gamma} / i^2}$.
By Lemma~2, we have $g_i \in \Hgamma$
and the $\Hgamma$ norm of $g_i$ is at most $i^{q/2}$ times its $\H_{\bm{\gamma} / i^2}$ norm.
Thus,
\[
\lVert f \rVert_{\H}
\leq 
\sum_{i=1}^{r} \binom{r}{i} \pi^{-q/4} \left( \prod^{q}_{j=1} \gamma_j\right)^{1/4} \lVert f_0 \rVert_2
\leq 2^r \pi^{-q/4} \big( \max_{j \in \S} \gamma_j \big)^{|\S|/4} \big( \max_{j \in \Sc} \gamma_j \big)^{|\Sc|/4} \lVert f_0 \rVert_2.
\]

Therefore, 
\begin{align*}
\lambda \lVert f \rVert^{2}_{\H} + \L(f) - \L(f_0)
&= \lambda \lVert f\rVert^2_{\H} 
+ \E \left\{ f(\bm{X}) - f_0(\bm{X}) \right\}^2 \\
&\leq 2^{2r} \pi^{-q/2} B^2 \lambda \big( \max_{j \in \S} \gamma_j \big)^{|\S|/2} \big( \max_{j \in \Sc} \gamma_j \big)^{|\Sc|/2}
+ c_1^2 c_2^2 \big(\min_{j \in \S} \gamma_j \big)^{-1}.
\end{align*}
In addition, for any $\bm{x} \in D$, it follows that 
\begin{align*}
|f(\bm{x})|
&\leq \sum^{r}_{i=1} \binom{r}{i} \int_{\R^q} \left(\frac{2}{\pi}\right)^{q/2}  \left( \prod^{q}_{j=1} \gamma_j \right)^{1/2} i^{-q} \exp \left\{- \sum^{q}_{j=1} 2\gamma_j(x_j - u_j)^2 / i^2\right\} d \bm{u} \cdot \lVert f_0 \rVert_\infty \\
&= \sum^{r}_{i=1} \binom{r}{i} \lVert f_0 \rVert_\infty
\leq 2^r B.
\end{align*}
\end{proof}

\subsection{Risk bounds for kernel ridge regression}

Recall that the truncation operator $\TB: L^{\infty}(D) \to L^{\infty}(D)$ is defined as
\[
\TB(f)(\bm{x}) = f(\bm{x}) I\{ -B \leq f(\bm{x}) \leq B \} + B I\{ f(\bm{x}) > B \} + (-B) I\{ f(\bm{x}) < -B \}, \quad \bm{x} \in D. 
\]
For any function $f$, $g$, we have $|\TB(f)(\bm{x}) - \TB(g) (\bm{x}) | \leq |f(\bm{x}) - g(\bm{x})|$. Hence, we have $\lVert \TB(f) - \TB(g) \rVert_{\infty} \leq \lVert f-g \rVert_{\infty}$. 
As a consequence, for any $B \geq \lVert f_0 \rVert_\infty$, we have
\[
\L\{\TB(f)\} - \L(f_0)
= \E \{ \TB(f)(\bm{X}) - f_0(\bm{X}) \}^2
\leq \E \{ f(\bm{X}) - f_0(\bm{X}) \}^2
= \L(f) - \L(f_0).
\]

Define $\wh{Y}_T = Y_T$ and $\wh{Y}_t = Y_t + \wh{Q}_{t+1} \{ \bm{X}_{t+1}, \wh{\pi}_{t+1} (\bm{X}_{t+1})\}$ for $t<T$.
Given sequences $\bm{\gamma}_n$ and $\lambda_n$, 
the estimator of the $Q$-function is 
$\wh{Q}_t(\cdot, a) = \TB(\wh{f}_n)$, where
\[
\wh{f}_n = \argmin_{f \in \Hgamma}
\Pn I(A=a) \{ \wh{Y} - f(\bm{X}) \}^2 + \lambda \lVert f \rVert_{\Hgamma}^2.
\]
To facilitate our analysis, we define
\[
d_n = \argmin_{f \in \Hgamma}
\Pn I(A=a) \{ \wt{Y} - f(\bm{X}) \}^2 + \lambda \lVert f \rVert_{\Hgamma}^2.
\]
Note that we omit the subscript $n$ in $\bm{\gamma}_n$ and $\lambda_n$
for simplicity.
The difference between $\wh{f}_n$ and $d_n$ is that we use
$\wt{Y}_t = Y_t + Q_{t+1} \{ \bm{X}_{t+1}, \pi^{\ast}_{t+1} (\bm{X}_{t+1})\}$ for $t<T$ when defining $\wh{f}_n$, which is an unobserved quantity since it relies on $\pi^{\ast}_{t+1}$ and $Q_{t+1}$.
In contrast, 
we replace $\pi^{\ast}_{t+1}$ and $Q_{t+1}$ by their estimates 
$\wh{\pi}_{t+1}$ and $\wh{Q}_{t+1}$ to obtain $\wh{Y}_t$. 
Hence $\wh{Q}_t(\cdot, a)$ is based on observed quantities only.

In this Section, we will show that the difference between
$\TB(\wh{f}_n)$ and $f_0 = Q_t(\cdot, a)$ is small.  To be precise,
define
$\mathcal{E}(f) = \lambda \lVert f \rVert^2_{\Hgamma} + \L \{ \TB (f)
\} - \L(f_0)$.
Our goal is to show that $\mathcal{E}(\wh{f}_n)$ is small with large
probability.  The proof below follows the idea in \citet[][Theorem
7.20]{steinwart2008support} while accounting for the error in the
responses.  For notational convenience, define
$\wb{\gamma}_{\S} = 1 + \max_{j \in \S} \gamma_j$,
$\wub{\gamma}_{\S} = \min_{j \in \S} \gamma_j$ and
$\wb{\gamma}_{\Sc} = 1 + \max_{j \in \Sc} \gamma_j$.  For any $f$,
define $\ell_f = I(A=a) \{ \wt{Y} - f( \bm{X} )\}^2$ and
$h_f = \ell_f - \ell_{f_0}$.  Then,
$\L(f) - \L(f_0) = \E h_{f}$.  Thus, $\E h_{f} \geq 0$ for all $f$.

\begin{lemma}
\label{thm:krr-decomp}
For any $f \in \Hgamma$, we have 
\[
\mathcal{E}(\wh{f}_{n}) \leq \lambda\lVert f \rVert^2_{\Hgamma} + \Pn h_f - \Pn h_{\wh{f}_n} + \E h_{\TB(\wh{f}_n)} + 2 \Pn \left( \wh{Y} - \wt{Y}\right)^2. 
\]
\end{lemma}

\begin{proof}
By the definition of $\wh{f}_n$ and $d_n$, we have 
\begin{align*}
\lambda \lVert \wh{f}_n \rVert^2_{\Hgamma} + \Pn I(A = a)\left\{\wh{Y} - \wh{f}_n(\bm{X})\right\}^2 &\leq \lambda \lVert d_n \rVert^2_{\Hgamma} + \Pn I(A=a) \left\{\wh{Y} - d_n(\bm{X})\right\}^2, \\
\lambda \lVert d_n \rVert^2_{\Hgamma} + \Pn I(A = a)\left\{\wt{Y} - d_n(\bm{X})\right\}^2 &\leq \lambda \lVert f \rVert^2_{\Hgamma} + \Pn I(A=a) \left\{\wh{Y} - f(\bm{X})\right\}^2 .
\end{align*}
Therefore,
\[
\lambda \lVert \wh{f}_n \rVert^2_{\Hgamma} \leq \lambda \lVert f \rVert^2_{\Hgamma} + \Pn h_f - \Pn h_{d_n} + \Pn I(A=a) \left\{ \wh{Y} - d_n(\bm{X})\right\}^2 - \Pn I(A=a) \left\{\wh{Y} - \wh{f}_n(\bm{X}) \right\}^2.
\]

For any real number $a_1$, $a_2$, $b_1$, $b_2$,  it follows that
\begin{align*}
(a_1 - b_1)^2 - (a_1 -b_2)^2 &= (2 a_1 - b_1 - b_2)(b_2 - b_1)\\
&= (2a_2 - b_1 -b_2) (b_2 -b_1) + 2(a_1 - a_2) (b_2 - b_1) \\
&\leq (a_2 - b_1)^2 - (a_2 - b_2)^2 + (a_1 -a_2 )^2 + (b_1 - b_2)^2.
\end{align*}
Hence,  
\begin{align*}
&\Pn I(A=a) \left\{ \wh{Y} - d_n(\bm{X}) \right\}^2 - \Pn I(A=a) \left\{\wh{Y} - \wh{f}_n(\bm{X}) \right\}^2\\
&\leq \Pn h_{d_n} - \Pn h_{\wh{f}_n} + \Pn I(A=a)\left( \wh{Y} - \wt{Y} \right)^2 + \Pn I(A=a) \left\{ \wh{f}_n(\bm{X}) - d_n(\bm{X})\right\}^2. 
\end{align*}

Let $\wh{\bm{Y}}$ be the vector of $\wh{{Y}}_i$, $i \in \mathcal{I}_a$, $\wt{\bm{Y}}$ the vector of $\wt{{Y}}_i$, $i \in \mathcal{I}_a$ and $\bm{K}$ the matrix of $K(\bm{X}_i, \bm{X}_j)$, $i, j \in \mathcal{I}_a$, where $\mathcal{I}_a = \{i : A_i = a\}$. 
By the representer theorem and the fact that all the eigenvalues of
$\bm{K}(\bm{K} + \lambda \bm{I})^{-1}$ are less than one, so that 
\[
\lVert \{\wh{f}_n(\bm{X}_i)\}_{i \in \mathcal{I}_a} - \left\{ d_n(\bm{X}_i)\right\}_{i \in \mathcal{I}_a} \rVert_2 
= \lVert \bm{K} (\bm{K} + \lambda \bm{I})^{-1} (\wh{\bm{Y}} - \wt{\bm{Y}})\rVert_2
\leq \lVert \wh{\bm{Y}} - \wt{\bm{Y}} \rVert_2.
\]
Thus, the inequality in the lemma follows from noting
\[
\Pn I(A=a) \left\{ \wh{f}_n(\bm{X}) - d_n(\bm{X})\right\}^2
\leq
\Pn I(A=a)\left( \wh{Y} - \wt{Y} \right)^2.
\]

\end{proof}

\begin{proposition}
\label{thm:krr-risk}
Suppose $\Pr\left\{\Pn (\wh{Y} - \wt{Y})^2 \geq c_1 n^{-\alpha} + c_2 n^{-\beta} \tau\right\} \leq e^{-\tau}$ for some $\alpha, \beta > 0$,
and $f_0$ satisfies the conditions in Proposition~\ref{thm:krr-approx}. 
Then for any $\delta > 0$ and $\tau > 0$, 
\begin{multline*}
\Pr \bigg[ \E_{\bm{X}}\left\{\TB(\wh{f}_n)(\bm{X}) - f_0(\bm{X}) \right\}^2 
\geq \\
c \left\{
\lambda \wb{\gamma}_{\S}^{|\S|/2} \bar{\gamma}_{\Sc}^{|\Sc|/2} 
+ \wub{\gamma}_{\S}^{-r}
+ \wb{\gamma}_{\S}^{|\S|/2} \wb{\gamma}_{\Sc}^{|\Sc|/2} \lambda^{-\delta} n^{-1} 
+ n^{-\alpha} 
+ n^{-\min(\beta, 1)}\tau \right\} \bigg] \leq e^{-\tau},
\end{multline*}
where $c$ is a constant that depends on $\delta$, $q$, $r$, $B$ and $\varpi$ only,
and $\E_{\bm{X}}$ denotes the expectation with respect to $\bm{X}$ only.
\end{proposition}

\begin{proof}

By Proposition~\ref{thm:krr-approx} and the inequality
$E \left[ I(A = a) \left\{ f(\bm{X}) - f_0(\bm{X}) \right\}^2 \right] \leq \lVert f - f_0 \rVert_\infty^2$, 
there exists some function $f_n \in \Hgamma$ such that
\begin{equation}
\label{eq:krr-term1}
\lambda \lVert f_n \rVert^{2}_{\Hgamma} + \E h_{f_n}
\leq c \left\{ \lambda\big( \max_{j \in \S} \gamma_j \big)^{|\S|/2} \big( \max_{j \in \Sc} \gamma_j \big)^{|\Sc|/2} 
+ \big( \min_{j \in \S} \gamma_j \big)^{-r} \right \}
\end{equation}
for some constant $c$ independent of $n$, and 
$\lVert f \rVert_{\infty} \leq 2^r B$.

By the property of the truncation operator and the fact that $\lVert Y \rVert_\infty \leq B$ with probability 1, we have
$\Pn h_{\TB(\wh{f}_n)} \leq \Pn h_{\wh{f}_n}$.
We apply Lemma~\ref{thm:krr-decomp} with $f = f_n$ to obtain
\begin{align*}
\mathcal{E}(\wh{f}_n)
&\leq \lambda \lVert f_n \rVert^2_{\Hgamma} + \Pn h_{f_n} - \Pn h_{\TB(\wh{f}_n)} + \E h_{\TB(\wh{f}_n)}
+ \Pn (\wh{Y} - \wt{Y})^2 \\
&\leq (\lambda \lVert f_n \rVert^2_{\Hgamma} + \E h_{f_n}) 
+ \lvert \Pn h_{f_n} - \E h_{f_n} \rvert
+ \lvert \E h_{\TB(\wh{f}_n)} - \Pn h_{\TB(\wh{f}_n)} \rvert
+ \Pn (\wh{Y} - \wt{Y})^2. 
\end{align*}
Note that $\E h_{\TB(\wh{f}_n)}$ is defined as computing $h_{\TB(f)}$ 
and then plugging in $f = \wh{f}_n$, thus
$\E h_{\TB(\wh{f}_n)}$ is a random variable.

We will consider the three terms  in the above display separately.
The first term can be bounded above using equation \eqref{eq:krr-term1}.

For the second term, we first observe that
\[
|h_{f_n}| \leq |\{ Y - f_n (\bm{X}) \}^2 - \{ Y - f_0(\bm{X})\}^2|
= |\{ f_n(\bm{X}) + f_0(\bm{X}) - 2Y\} \{ f_n(\bm{X}) - f_0(\bm{X})\}|.
\]
Because  $\lVert f_0 \rVert_{\infty} \leq \wt{B}$ and  $\lVert f_n \rVert_{\infty} \leq \wt{B} $ for $\wt{B} = 2^r B$,
we have $E h^2_{f_n} \leq 16 \wt{B}^2 \E\{f_n(\bm{X}) - f_0(\bm{X})\}^2 = 16 \wt{B}^2 \E h_{f_n}$ and $| h_{f_n} | \leq 8 \wt{B}^2$. 
By Bernstein's inequality \citep[][Theorem 6.12]{steinwart2008support}, 
we obtain
\[
\Pr \left( |\Pn h_{f_n} - \E h_{f_n}| \geq \frac{16 \wt{B}^2 \tau}{3n}+ \left\{\frac{32 \wt{B}^2 \tau ( \E h_{f_n})}{n}\right\}^{1/2} \right) \leq 2 e^{-\tau}.
\]
Using $2 (u v)^{1/2} \leq u + v$, it follows that 
\[
\left\{\frac{32 \wt{B}^2 \tau ( \E h_{f_n})}{n}\right\}^{1/2}
\leq \frac{8 \wt{B}^2 \tau}{n} + \E h_{f_n}
\leq \frac{8 \wt{B}^2 \tau}{n} + \E h_{f_n} + \lambda \lVert f_n \rVert^2_{\Hgamma}.
\]
Therefore, 
\begin{equation}
\label{eq:krr-term2}
\Pr \left(  \left| \Pn h_{f_n} - \E h_{f_n} \right| \geq \frac{14 \wt{B}^2 \tau}{n} +  \E h_{f_n} + \lambda \lVert f_n \rVert^2_{\Hgamma} \right) \leq 2 e^{-\tau}.
\end{equation}

Bounding the third term is a little bit more involved.
Let $s > 0$ be fixed;  for any $f \in \Hgamma$, define 
\[
m_f = \frac{h_{\TB(f)} - \E h_{\TB(f)}}{\mathcal{E}(f) + s} 
= \frac{h_{\TB(f)} - \E h_{\TB(f)}}{\lambda \lVert f \rVert^2_{\Hgamma} + \E h_{\TB(f)} + s}.
\]
Because $\lVert \TB(f) \rVert_{\infty} \leq B$,  
$\lVert m_f\rVert_{\infty} \leq 16B^2/s$. 
Furthermore, because $\E h^2_{\TB(f)} \leq 16B^2 \E h_{\TB(f)}$, 
\[
\E m^2_f \leq \frac{\E h^2_{\TB(f)}}{4s\E h_{\TB(f)}} \leq \frac{4 B^2}{s}, 
\]
where $\E h_{\TB(f)} > 0$, and 
$
\E h^2_{\TB(f)} =0 \leq 4B^2/s
$
when $\E h_{\TB(f)} =0$. 

Define $\mathcal{F}_s = \{f \in \Hgamma: \mathcal{E}(f) \leq s\} \cup \{0\}$, where $0$ denotes the zero function. 
By Corollary~\ref{thm:talagrand},  it follows that 
\[
\Pr \left\{ \sup_{f \in \mathcal{F}_s} |m_f| \geq 2 \E \sup_{f \in \mathcal{F}_s} |m_f| + \left( \frac{8B^2\tau}{ns} \right)^{1/2} + \frac{32 B^2 \tau}{ns} \right\} \leq e^{-\tau}.
\]
We shall derive an upper bound for $\sup_{f \in \mathcal{F}_s} |m_f|$
based on an upper bound for $\E \sup_{f \in \mathcal{F}_s} |h_{\TB(f)} - \E h_{\TB(f)}|$.
To this end, we compute the covering number for $\mathcal{G}_s = \{ h_{\TB(f)} - \E h_{\TB(f)}: f \in \mathcal{F}_s \}$.

For any $f \in \mathcal{F}_s$, we have $\lVert f \rVert_{\Hgamma} \leq s^{1/2} \lambda^{-1/2}$. Hence,
\[
\mathcal{N}(\mathcal{F}_s, \lVert \cdot \rVert_{\infty}, \varepsilon) 
\leq \mathcal{N} \{ (s^{1/2} \lambda^{-1/2}) \mathcal{B}_{\Hgamma}, \lVert \cdot \rVert_{\infty}, \varepsilon\} 
= \mathcal{N}(\mathcal{B}_{\Hgamma}, \lVert \cdot \rVert_{\infty}, s^{-1/2}\lambda^{1/2}\varepsilon). 
\]
By the fact that $\lVert \TB(f) - \TB(g) \rVert_{\infty} \leq \lVert f-g \rVert_{\infty}$, 
\[
\lVert h_{\TB(f)} - \E h_{\TB(f)} - h_{\TB(g)} + \E h_{\TB(g)} \rVert_{\infty} \leq 8B \lVert f-g \rVert_{\infty}.
\]
Hence, $\mathcal{N} (\mathcal{G}_s, \lVert \cdot \rVert_{\infty}, \varepsilon) \leq \mathcal{N}\{ \mathcal{F}_s, \lVert \cdot \rVert_{\infty}, \varepsilon/(8B) \}.$
Combining these inequalities and applying
Lemma~\ref{thm:rkhs-covering}, shows 
\[
\log \mathcal{N} ( \mathcal{G}_s, \lVert \cdot \rVert_{\infty}, \varepsilon) 
\leq \log \mathcal{N} \{ \mathcal{B}_{\Hgamma}, \lVert \cdot \rVert_{\infty}, (8B)^{-1} s^{-1/2}\lambda^{1/2}\varepsilon\}
\leq c_1 a_{\bm{\gamma}} s^{q/(2m)}\lambda^{-q/(2m)}\varepsilon^{-q/m},
\]
where $m \geq 1$ is an arbitrary integer, 
$c_1$ is a constant that depends on $m$, $q$, $B$, $r$ only, 
and $a_{\bm{\gamma}} = \prod^{q}_{j=1} (1 + \gamma_j)^{1/2}
\leq \wb{\gamma}_{\S}^{|\S|/2} \wb{\gamma}_{\Sc}^{|\Sc|/2}$.

For any $f \in \mathcal{F}_s$, we have 
$\lVert h_{\TB(f)} - \E h_{\TB(f)} \rVert_{\infty} \leq 16 B^2$ and 
$\var h_{\TB(f)} \leq \E h^2_{\TB(f)} \leq 16 B^2s$. 
Apply Proposition~\ref{thm:sup-tala} to obtain 
\[
\E \sup_{f \in \mathcal{F}_s} |h_{\TB(f)} - \E h_{\TB(f)}| \leq 1024 (16 B^2 J / n) + 64 (16 B^2 J s / n)^{1/2},
\]
where $J = \int^{1}_{0} c_1 a_{\bm{\gamma}} (16 B^2)^{q/(2m)} 
\lambda^{-q/(2m)} \varepsilon^{-q/m} \, d\varepsilon \leq c_2 a_{\bm{\gamma}} \lambda^{-q/(2m)}$. 
Thus, 
\[
\E \sup_{f \in \mathcal{F}_s} |h_{\TB(f)} - \E h_{\TB(f)}| \leq c_3 \left\{ a_{\bm{\gamma}} \lambda^{-q/(2m)}n^{-1} + a_{\bm{\gamma}}^{1/2} \lambda^{-q/(4m)} s^{1/2} n^{-1/2} \right\}.
\]
Hence, by the peeling technique \citep[][Theorem 7.7]{steinwart2008support}, 
we obtain
\[
\E \sup_{f \in \mathcal{F}} |m_f| \leq 4 c_3 \{ a_{\bm{\gamma}} \lambda^{-q/(2m)} s^{-1}n^{-1} + a_{\bm{\gamma}}^{1/2} \lambda^{-q/(4m)} s^{-1/2} n^{-1/2} \}. 
\]
Combine the bound of $\E \sup_{f \in \mathcal{F}} |m_f|$
and the tail bound of $\sup_{f \in \mathcal{F}} |m_f|$ to obtain
\[
\Pr \left[ \sup_{f \in \mathcal{F}} \frac{|h_{\TB(f)} - \E h_{\TB(f)}|}{\mathcal{E}(f) + s} 
\geq c_4 \left\{ \frac{a_{\bm{\gamma}}}{\lambda^{q/(2m)}sn} + \frac{a^{1/2}_{\bm{\gamma}}}{\lambda^{q/(4m)}s^{1/2}n^{1/2}} + \frac{\tau^{1/2}}{s^{1/2}n^{1/2}} + \frac{\tau}{sn} \right\}  \right] \leq e^{-\tau},
\]
where $c_4 > 0$ is some constant that depends on $m$, $q$, $B$, $r$ only. 
Without loss of generality, we assume $c_4 \geq 1$.

Let 
\[
s = 64 c_4^2 \max \left\{ \frac{a_{\bm{\gamma}}}{\lambda^{q/(2m)}n}, \frac{\tau}{n} \right\},
\]
then 
\[
\frac{c^2_4 a_{\bm{\gamma}}}{\lambda^{q/(2m)}sn} \leq \left( \frac{c_4^2 a_{\bm{\gamma}}}{\lambda^{q/(2m)}sn} \right)^{1/2} \leq \frac{1}{8},
\quad
\frac{c_4^2 \tau}{sn} \leq \left( \frac{c^2_4 \tau}{sn} \right)^{1/2} \leq \frac{1}{8}.
\]
Therefore, we have 
\begin{equation}
\label{eq:krr-term3}
\Pr \left\{ |\Pn h_{\TB(f)} - \E h_{\TB(f)}| \geq \mathcal{E}(f) / 2 + s / 2
\text{ for some } f \in \mathcal{F} \right\} \leq e^{-\tau}.
\end{equation}

Plug-in $f = \wh{f}_n$ in equation~\eqref{eq:krr-term3} and combine
equations~\eqref{eq:krr-term1}, \eqref{eq:krr-term2},
\eqref{eq:krr-term3} and the condition on $\Pn (\wh{Y} -
\wt{Y})^2$ to  obtain
\[
\Pr \left\{ \mathcal{E}(\wh{f}_n) \leq c_6 \left( \lambda \wb{\gamma}_{\S}^{|\S|/2} \bar{\gamma}_{\Sc}^{|\Sc|/2} + \wub{\gamma}_{\S}^{-r}+ \wb{\gamma}_{\S}^{|\S|/2} \wb{\gamma}_{\Sc}^{|\Sc|/2} \lambda^{-q/(2m)} n^{-1} + n^{-1} \tau + 
n^{-\alpha} + n^{-\beta}\tau \right) \right\} \leq e^{-\tau}.
\]
Since $m$ can be arbitrarily large, $\delta = q / (2m)$ can be arbitrarily small.

Finally, 
since
\[
\E_{\bm{X}} I(A = a) \left\{ \TB(\wh{f}_n)(\bm{X}) - f_0(\bm{X}) \right\}^2 
\leq \mathcal{E}\left( \wh{f}_n \right),
\]
we observe that
\begin{align*}
\E_{\bm{X}} I(A=a) \{ \TB(\wh{f}_n)(\bm{X}) - f_0(\bm{X}) \}^2 
&= \E_{\bm{X}} \Pr(A=a| \bm{X} ) \{ \TB(\wh{f}_n)(\bm{X}) - f_0(\bm{X})\}^2 \\
&\geq \varpi \E_{\bm{X}} \left\{\TB(\wh{f}_n)(\bm{X}) - f_0(\bm{X}) \right\}^2
\end{align*}
by Assumption~2.
\end{proof}

We immediately obtain the following corollaries.

\begin{corollary}
Assume the conditions in Proposition~\ref{thm:krr-risk} hold.
Furthermore, 
suppose $\wb{\gamma}_{\S} = \wb{\theta}_{\S} n^{2/(2r + q)}$, 
$\wub{\gamma}_{\S} = \wub{\theta}_{\S} n^{2/(2r + q)}$, 
$\wb{\gamma}_{\Sc} = \wb{\theta}_{\Sc} n^{2/(2r + q)}$, 
and $\lambda = \theta_{\lambda}n^{-1}$. 
Then, for any $\xi > 0$, 
\[
\Pr \left( \E_{\bm{X}} \left\{\TB(\wh{f}_n)(\bm{X}) - f_0(\bm{X}) \right\}^2  \geq 
c \left[ n^{-\min\{2r/(2r+q) + \xi, \alpha\}} 
+ n^{-\min(\beta, 1)} \tau \right] \right)
\leq e^{-\tau},
\]
\end{corollary}

\begin{corollary}
Assume the conditions in Proposition~\ref{thm:krr-risk} hold.
Furthermore, 
suppose $\wb{\gamma}_{\S} = \wb{\theta}_{\S} n^{2/(2r+|\S|)}$, 
$\wub{\gamma}_{\S} = \wub{\theta}_{\S} n^{2/(2r+|\S|)}$, 
$\wb{\gamma}_{\Sc} = \wb{\theta}_{\Sc}$, 
and $\lambda = \theta_{\lambda} n^{-1}$.
Then, for any $\xi>0$, 
\[
\Pr \left( \E_{\bm{X}} \left\{\TB(\wh{f}_n)(\bm{X}) - f_0(\bm{X}) \right\}^2  \geq 
c \left[ n^{-\min\{2r/(2r + |\S|) + \xi, \alpha\}} 
+ n^{-\min(\beta, 1)} \tau \right] \right)
\leq e^{-\tau},
\]
\end{corollary}

\subsection{Useful inequalities for the analysis of decision lists}

Define $U_t(\bm{x}, a) = \max_{a' \in \mathcal{A}_t} Q_t(\bm{x}, a') - Q_t(\bm{x}, a)$ and $\wh{U}_t(\bm{x}, a) = \max_{a' \in \mathcal{A}_t} \wh{Q}_t(\bm{x}, a') - Q_t(\bm{x}, a)$. 
Because
\[
\left|\max_{a' \in \mathcal{A}_t} \wh{Q}_t(\bm{x}, a') - \max_{a' \in \mathcal{A}_t} Q_t(\bm{x}, a)\right| \leq \max_{a' \in \mathcal{A}_t} 
\left|\wh{Q}_t(\bm{x}, a') - Q_t(\bm{x}, a)\right|,
\]
it follows that
\[
\left|\wh{U}_t(\bm{x}, a) - U_t(\bm{x}, a)\right|
\leq 2 \max_{a' \in \mathcal{A}_t} \left|\wh{Q}_t(\bm{x}, a') - Q_t(\bm{x}, a')\right|.
\]
Thus, for any $p \geq 1$
\begin{equation}
\label{eq:u-q}
\Pn \left| \wh{U}_t(\bm{X_t}, a) - U_t(\bm{X_t}, a) \right|^p
\leq 2 \sum_{a' \in \mathcal{A}_t} \left| \wh{Q}_t(\bm{X_t}, a') - Q_t(\bm{X_t}, a') \right|^p.
\end{equation}

Following the notation used in the algorithm description, define 
\[
\wh{\Omega}_{t\ell} (R, a) = I( \bm{X}_t \in \wh{G}_{t\ell}, \bm{X}_t \in R ) \left\{ \wh{U}_t(\bm{X}_t, a) - \zeta \right\} - \eta \left\{ 2-V(R) \right\}
\]
and 
\[
\Omega_{t\ell}(R, a) = I ( \bm{X}_t \in G^{\ast}_{t\ell}, \bm{X}_t \in R ) \left\{ U_t(\bm{X}_t, a) - \zeta \right\} - \eta \left\{ 2-V(R) \right\}.
\]
By the definition of $( \wh{R}_{t\ell}, \wh{a}_{t\ell} )$ in the main article, we have
\begin{multline*}
( \wh{R}_{t\ell}, \wh{a}_{t\ell} ) 
= \argmax_{R \in \mathcal{R}_t, a \in \mathcal{A}_t} 
\Pn I(\bm{X}_t \in \wh{G}_{t\ell}) \wh{Q}_t \{ \bm{X}_t, \wh{\pi}^Q_t(\bm{X}_t) \} \\
- \Pn I(\bm{X}_t \in \wh{G}_{t\ell}, \bm{X}_t \in R) \wh{U}_t(\bm{X}_t, a)  \\
+ \Pn\zeta I\{\bm{X}_t \in \wh{G}_{t\ell}, \bm{X}_t \in R\} + \eta\{2 - V(R)\}.
\end{multline*}
Thus,  $( \wh{R}_{t\ell}, \wh{a}_{t\ell} ) =
\argmin_{R \in \mathcal{R}_t, a \in \mathcal{A}_t} 
\Pn \wh{\Omega}_{t\ell} (R, a)$.
Similarly, 
\[
\left( R^{\ast}_{t\ell}, a^{\ast}_{t\ell} \right) 
= \argmax_{R \in \mathcal{R}_t, a \in \mathcal{A}_t} \Psi_{t\ell}(R, a) 
= \argmin_{R \in \mathcal{R}_t, a \in \mathcal{A}_t} \E \Omega_{t\ell}(R, a).
\]

Recall that $\mathcal{R}_t$ consists of rectangles in $\R^q$ defined using at most two variables. Hence $\mathcal{R}_t$ is a subset of the set of all intervals $\{ (\bm{a}, \bm{b}]: \bm{a}, \bm{b} \in \R^q \}$, where $(\bm{a}, \bm{b}] = \{\bm{x} \in \R^q: a_j \leq x_j \leq b_j \text{ for all } j\}$. Hence $\mathcal{R}_t$ is a Vapnik-Cervonenkis class, or VC class for short \citep[][Example 2.6.1]{van1996weak}.

The following lemma gives an upper bound 
for $\sup_{R \in \mathcal{R}_t} |\Pn \wh{\Omega}_{t\ell}(R, a) - \E \Omega_{t\ell} (R, a)|$ for any given $a \in \mathcal{A}_t$. 

\begin{lemma}
\label{thm:cr-global}
We have
\begin{multline*}
\sup_{R \in \mathcal{R}_t} \left|\Pn \wh{\Omega}_{t\ell} (R, a) - \E \Omega_{t\ell}(R, a)\right| 
\leq \sup_{R \in \mathcal{R}_t} \left|\Pn \Omega_{t\ell} (R, a) - \E \Omega_{t\ell} (R, a)\right| \\
+ \Pn \left| \wh{U}_t(\bm{X}_t, a) - U_t(\bm{X}_t, a) \right|
+ B \sum_{k < \ell} \Pn I\left(\bm{X}_t \in \wh{R}_{t k} \symmdiff R^{\ast}_{t k} \right),
\end{multline*}
and
\[
\pr\left\{
\sup_{R \in \mathcal{R}_t} |\Pn \Omega_{t\ell} (R, a) - \E \Omega_{t\ell} (R, a)| \geq c(n^{-1/2} + \tau^{1/2} n^{-1/2})
\right\}
\leq e^{-\tau}.
\]
\end{lemma}

\begin{proof}
We have
\begin{multline*}
\sup_{R \in \mathcal{R}_t} |\Pn \wh{\Omega}_{t\ell} (R, a) - \E \Omega_{t\ell}(R, a)|  \\
\leq \sup_{R \in \mathcal{R}_t} |\Pn \wh{\Omega}_{t\ell} (R, a) - \Pn \Omega_{t\ell} (R, a)| + \sup_{R \in \mathcal{R}_t} |\Pn \Omega_{t\ell} (R, a) - \E \Omega_{t\ell} (R, a)|.
\end{multline*}
For the first term, we observe that
\begin{align*}
\sup_{R \in \mathcal{R}_t} & \left| \Pn \wh{\Omega}_{t \ell} (R, a) - \Pn \Omega_{t\ell}(R, a)\right| \\
&\leq \sup_{R \in \mathcal{R}_t} \left| \Pn I(\bm{X}_t \in \wh{G}_{t\ell} \cap G^{\ast}_{t\ell}, \bm{X}_t \in R) \left\{ \wh{U}_t (\bm{X}_t, a) - U_t(\bm{X}_t, a) \right\}  \right| \\
&\quad {} + \sup_{R \in \mathcal{R}_t} \Pn I(\bm{X}_t \in \wh{G}_{t\ell} \symmdiff G^{\ast}_{t\ell}, \bm{X}_t \in R) \left| \wh{U}_t(\bm{X}_t, a) - \zeta\right| \\
&\leq \Pn \left| \wh{U}_t (\bm{X}_t, a) - U_t(\bm{X}_t, a)\right| 
+ B \Pn I(\bm{X}_t \in \wh{G}_{t\ell} \symmdiff G^{\ast}_{t\ell}),
\end{align*}
By the definition of $G_{t\ell}$, 
\[
\Pn I(\bm{X}_t \in \wh{G}_{t\ell} \symmdiff G^{\ast}_{t\ell})
\leq \sum_{k < \ell} \Pn I(\bm{X}_t \in \wh{R}_{t k} \symmdiff R^{\ast}_{t k}).
\]

For the second term, VC preservation properties \citep[][Lemma 2.6.18]{van1996weak}, 
the set 
\[
\mathcal{F} = \{ I (\bm{X}_t \in R) I(\bm{X}_t \in G^{\ast}_{t\ell}) \{U_t (\bm{X}_t, a) - \zeta\}: R \in \mathcal{R}_t \}
\]
is also a VC class. 
Let $\nu$ be its VC index. Then, by Theorem 2.6.7 in \citet[][]{van1996weak}, 
\[
\sup_{Q} \mathcal{N}(\mathcal{G}, \lVert \cdot \rVert_{L^{2}(Q)}, \varepsilon) \leq c_1 \varepsilon^{-2\nu},
\]
where $Q$ is any probability measure and $c_1$ is a constant that
depends on $\nu$ only.  For any $f \in \mathcal{F}$, it can be seen
that 
$\lVert f \rVert_{\infty} \leq B$.  Thus, by
Propositions~\ref{thm:hoeffding} and~\ref{thm:sup-hoef}, since
$\int_0^1 \log (\varepsilon^{-2\nu}) < \infty$, we have
\[
\Pr \left\{ \sup_{f \in \mathcal{F}} |\Pn f - \E f| \geq c \left( \frac{B^2}{n} \right)^{1/2} + c\left( \frac{B^2\tau}{n} \right)^{1/2}  \right\} \leq e^{-\tau}
\]
for any $\tau > 0$, where $c$ is some constant that depends on $\nu$. 

\end{proof}

Recall that $\rho_t(R_1, R_2) = \Pr (\bm{X}_t \in R_1 \symmdiff
R_2)$. The following lemma gives an upper bound  on 
\[
\sup_{R \in \mathcal{R}_t, \rho_t(R, R^{\ast}_{t\ell}) \leq \delta} 
\left|
\left\{ 
\Pn \wh{\Omega}_{t\ell}(R, a^{\ast}_{t\ell}) - \E \Omega_{t\ell} (R, a^{\ast}_{t\ell})
\right\} 
-
\left\{ 
\Pn \wh{\Omega}_{t\ell}(R^{\ast}_{t\ell}, a^{\ast}_{t\ell}) - \E \Omega_{t\ell} (R^{\ast}_{t\ell}, a^{\ast}_{t\ell})
\right\} 
\right|. 
\]

\begin{lemma}
\label{thm:cr-local}
We have
\begin{align*}
& \sup_{R \in \mathcal{R}_t, \rho_t(R, R^{\ast}_{t\ell}) \leq \delta} 
\left|
\left\{ 
\Pn \wh{\Omega}_{t\ell}(R, a^{\ast}_{t\ell}) - \E \Omega_{t\ell} (R, a^{\ast}_{t\ell})
\right\} 
-
\left\{ 
\Pn \wh{\Omega}_{t\ell}(R^{\ast}_{t\ell}, a^{\ast}_{t\ell}) - \E \Omega_{t\ell} (R^{\ast}_{t\ell}, a^{\ast}_{t\ell})
\right\} 
\right| \\
\leq 
& \sup_{R \in \mathcal{R}_t, \rho_t(R, R^{\ast}_{t\ell}) \leq \delta} 
\left|
\left\{ 
\Pn \Omega_{t\ell}(R, a^{\ast}_{t\ell}) - \E \Omega_{t\ell} (R, a^{\ast}_{t\ell})
\right\} 
-
\left\{ 
\Pn \Omega_{t\ell}(R^{\ast}_{t\ell}, a^{\ast}_{t\ell}) - \E \Omega_{t\ell} (R^{\ast}_{t\ell}, a^{\ast}_{t\ell})
\right\} 
\right| \\
& {}+ 
\left\{ \sup_{R \in \mathcal{R}_t, \rho_t(R, R^{\ast}_{t\ell}) \leq \delta} 
\Pn I(\bm{X}_t \in R \symmdiff R^{\ast}_{t\ell}) \right\}^{1/2} \left[ \Pn \left\{ \wh{U}_t(\bm{X}_t, a) - U_t(\bm{X}_t, a) \right\}^2 \right]^{1/2} \\
& {}+ B \left\{ \sup_{R \in \mathcal{R}_t, \rho_t(R, R^{\ast}_{t\ell}) \leq \delta} 
\Pn I(\bm{X}_t \in R \symmdiff R^{\ast}_{t\ell}) \right\}^{1/2} \left\{ \sum_{k < \ell} \Pn I\left(
\bm{X}_t \in \wh{R}_{t k} \symmdiff R^{\ast}_{t k} \right) \right\}^{1/2}.
\end{align*}
Let $J$ denote the first term on the right hand side of the above equation. We have
\[
\pr\left\{J \geq c \delta^{1/2 - \beta} (n^{-1/2} + n^{-1/2}\tau^{1/2}) \right\} \leq e^{-\tau}.
\]
In addition, 
\[
\pr\left\{ \sup_{R \in \mathcal{R}_t, \rho_t(R, R^{\ast}_{t\ell})}
\Pn I(\bm{X}_t \in R \symmdiff R^\ast_{t\ell}) \geq c \delta^{1-\beta} (1 + n^{-1} \tau) \right \} \leq e^{-\tau}.
\]
\end{lemma}

\begin{proof}

We have
\begin{align*}
& \sup_{R \in \mathcal{R}_t, \rho_t(R, R^{\ast}_{t\ell}) \leq \delta} 
\left|
\left\{ 
\Pn \wh{\Omega}_{t\ell}(R, a^{\ast}_{t\ell}) - \E \Omega_{t\ell} (R, a^{\ast}_{t\ell})
\right\} 
-
\left\{ 
\Pn \wh{\Omega}_{t\ell}(R^{\ast}_{t\ell}, a^{\ast}_{t\ell}) - \E \Omega_{t\ell} (R^{\ast}_{t\ell}, a^{\ast}_{t\ell})
\right\} 
\right| \\
\leq 
& \sup_{R \in \mathcal{R}_t, \rho_t(R, R^{\ast}_{t\ell}) \leq \delta} 
\left|
\left\{ 
\Pn \Omega_{t\ell}(R, a^{\ast}_{t\ell}) - \E \Omega_{t\ell} (R, a^{\ast}_{t\ell})
\right\} 
-
\left\{ 
\Pn \Omega_{t\ell}(R^{\ast}_{t\ell}, a^{\ast}_{t\ell}) - \E \Omega_{t\ell} (R^{\ast}_{t\ell}, a^{\ast}_{t\ell})
\right\} 
\right| \\
& {}+ \sup_{R \in \mathcal{R}_t, \rho_t(R, R^{\ast}_{t\ell}) \leq \delta}  
\left| \Pn \wh{\Omega}_{t\ell}(R, a) - \Pn \wh{\Omega}_{t\ell} (R^{\ast}_{t\ell}, a) - \Pn \Omega_{t\ell}(R, a) + \Pn \Omega_{t\ell} (R^{\ast}_{t\ell}, a) \right|.
\end{align*}

The first term can be bounded above using properties of VC classes.
For any $\delta > 0$, define 
\begin{multline*}
\mathcal{F}_{\delta} = \big\{ I(\bm{X}_t \in R) I(\bm{X}_t \in G^{\ast}_{t\ell}) \left\{ U_t(\bm{X}_t, a) - \zeta \right\} \\
- I(\bm{X}_t \in R^{\ast}_{t\ell}) I(\bm{X}_t \in G^{\ast}_{t\ell}) \left\{ U_t(\bm{X}_t, a) - \zeta \right\}: R \in \mathcal{R}_t, \rho_t(R, R^{\ast}_{t\ell}) \leq \delta \big\}.
\end{multline*}
Because $\mathcal{R}_t$ is a VC class,
$\mathcal{F}_\delta$ is a VC class for any $\delta$. 
In addition, 
$
\sup_{Q} \mathcal{N}(\mathcal{F}_\delta, \lVert \cdot \rVert_{L^2(Q)}, \varepsilon) \leq c_1 \varepsilon^{-2\nu}
$
for some constants $c_1$ and $\nu$
independent of $\delta$.

For any $f \in \mathcal{F}_{\delta}$, we have $\lVert f \rVert_{\infty} \leq B$ and $\E f^2 \leq B^2 \delta$. Thus, by Propositions~1 and~3, 
\[
\Pr \left[ \sup_{f \in \mathcal{F}_{\delta}} |\Pn f - \E f| \geq 
c_2 \left\{
\frac{\delta^{1/2} \log^{1/2} (1/\delta)}{n^{1/2}} + 
\frac{\log (1/\delta)}{n} + 
\frac{\delta^{1/2} \tau^{1/2}}{n^{1/2}}+ 
\frac{\tau^{1/2}}{n^{1/2}} \right\}  \right] \leq e^{-\tau} ,
\]
where $c_2$ is some constant that depends on $B$. 
As $\delta \in (0, 1]$,  it follows that 
$\log (1/\delta) \leq c_3 \delta^{-\beta}$ for any $\beta > 0$, 
where $c_3$ is same constant that depends on $\beta$ only. 
Thus, 
\[
\Pr \left\{ \sup_{f \in \mathcal{F}_{\delta}} |\Pn f - \E f| \geq c_4 \delta^{1/2 - \beta} \left( \frac{1}{n^{1/2}} + \frac{1}{n\delta^{1/2}} + \frac{\delta^{\beta} \tau^{1/2}}{n^{1/2}} + \frac{\tau}{n\delta^{1/2}} \right) \right\} \leq e^{-\tau}. 
\]

Hence, when $\delta^{1/2} \geq n^{-1/2}\tau^{1/2}$, we have 
\[
\Pr \left\{ \sup_{f \in \mathcal{F}_{\delta}} |\Pn f - \E f| \geq c_5 \delta^{1/2 - \beta} \left( n^{-1/2} + n^{-1/2} \tau^{1/2} \right)\right\} \leq e^{-\tau}.
\]

For the second term, we observe that
\begin{align*}
&\left| \Pn \wh{\Omega}_{t\ell}(R, a) - \Pn \wh{\Omega}_{t\ell} (R^{\ast}_{t\ell}, a) - \Pn \Omega_{t\ell}(R, a) + \Pn \Omega_{t\ell} (R^{\ast}_{t\ell}, a) \right| \\
={} &\Big| \Pn \left\{ I(\bm{X}_t \in R) - I(\bm{X}_t \in R^{\ast}_{t\ell})\right\} I(\bm{X}_t \in  \wh{G}_{t\ell}) \left\{ \wh{U}_t(\bm{X}_t, a) - \zeta \right\}  \\
&- \Pn \left\{ I(\bm{X}_t \in R) - I(\bm{X}_t \in R^{\ast}_{t\ell})\right\} I(\bm{X}_t \in  {G}^{\ast}_{t\ell}) \left\{ \wh{U}_t(\bm{X}_t, a) - \zeta  \right\}  \Big| \\
\leq {} &\Pn I(\bm{X}_t \in R \symmdiff R^{\ast}_{t\ell}) I(\bm{X}_t \in \wh{G}_{t\ell} \cap G^{\ast}_{t\ell}) \left| \wh{U}_t(\bm{X}_t, a) - U_t(\bm{X}_t, a) \right| \\
&{}+ \Pn I(\bm{X}_t \in R \symmdiff R^{\ast}_{t\ell}) I (\bm{X}_t \in \wh{G}_{t\ell} \symmdiff G^{\ast}_{t\ell}) B.
\end{align*}
Using the Cauchy-Schwarz inequality, 
\begin{align*}
&\Pn I(\bm{X}_t \in R \symmdiff R^{\ast}_{t\ell}) I(\bm{X}_t \in \wh{G}_{t\ell} \cap G^{\ast}_{t\ell}) \left| \wh{U}_t(\bm{X}_t, a) - U_t(\bm{X}_t, a) \right| \\
\leq &\Pn I(\bm{X}_t \in R \symmdiff R^{\ast}_{t\ell}) \left| \wh{U}_t(\bm{X}_t, a)- U_t(\bm{X}_t, a)\right| \\
\leq &\left\{ \Pn I(\bm{X}_t \in R \symmdiff R^{\ast}_{t\ell})\right\}^{1/2} \left[\Pn\left\{ \wh{U}_{t}(\bm{X}_t, a) - U_t(\bm{X}_t, a)\right\}^2 \right]^{1/2},
\end{align*}
and
\begin{align*}
&\Pn I(\bm{X}_t \in R \symmdiff R^{\ast}_{t\ell}) I(\bm{X}_t \in \wh{G}_{t\ell} \symmdiff G^{\ast}_{t\ell}) B \\
\leq &B \left\{ \Pn I(\bm{X}_t \in R \symmdiff R^{\ast}_{t\ell})\right\}^{1/2} \left\{ \Pn I(\bm{X}_t \in \wh{G}_{t\ell} \symmdiff G^{\ast}_{t\ell}) \right\}^{1/2}.
\end{align*}

Therefore, 
\begin{align*}
&\sup_{R \in \mathcal{R}_t, \rho_t(R, R^{\ast}_{t\ell})}
 \left| \Pn \wh{\Omega}_{t\ell}(R, a) - \Pn\wh{\Omega}_{t\ell} (R^{\ast}_{t\ell}, a) - \Pn \Omega_{t\ell}(R, a) + \Pn\Omega_{t\ell}(R^{\ast}_{t\ell}, a) \right| \\
\leq &\left\{ \sup_{R \in \mathcal{R}_t, \rho_t(R, R^{\ast}_{t\ell})}
 \Pn I(\bm{X}_t \in R \symmdiff R^{\ast}_{t\ell}) \right\}^{1/2} \left[ \Pn \left\{ \wh{U}_t(\bm{X}_t, a) - U_t(\bm{X}_t, a) \right\}^2 \right]^{1/2} \\
&+ B \left\{ \sup_{R \in \mathcal{R}_t, \rho_t(R, R^{\ast}_{t\ell})}
 \Pn I(\bm{X}_t \in R \symmdiff R^{\ast}_{t\ell}) \right\}^{1/2} \left\{ \Pn I\left(\bm{X}_t \in \wh{G}_{t\ell} \symmdiff G^{\ast}_{t\ell}\right) \right\}^{1/2}.
\end{align*}

Finally, let  $\mathcal{G}_{\delta} = \left\{ I(\bm{X}_t \in R \symmdiff R^{\ast}_{t\ell}): R \in \mathcal{R}_{t}, \rho_t (R , R^{\ast}_{t\ell}) \leq \delta \right\} $. 
Then, for any $g \in \mathcal{G}_{\delta}$,  $\lVert g \rVert_{\infty} \leq 1$ and $\E g^2 \leq \delta$. 
Thus, by Propositions~1 and~3, 
\[
\Pr \left[ \sup_{g \in \mathcal{G}_{\delta}} |\Pn g - \E g| \geq c_6
\left\{
\frac{\delta^{1/2} \log^{1/2}(1/\delta)}{n^{1/2}}
+ \frac{\log(1/\delta)}{n} 
+ \frac{\delta^{1/2} \tau^{1/2}}{n^{1/2}}
+ \frac{\tau}{n} \right\} \right] \leq e^{-\tau}
\]
Because $\sup_{g \in \mathcal{G}_{\delta}} \E g \leq \delta$ and
$(\delta/n)^{1/2} \leq (\delta + 1/n)/2$, 
\[
\Pr \left\{ \sup_{g \in \mathcal{G}_{\delta}}\Pn g \geq c_7 \delta^{1-\beta} \left( 1 + \frac{\tau}{n}\right) \right\} \leq e^{-\tau}.
\]

\end{proof}

The following lemma is useful for establishing the rate of convergence.
It is a finite-sample version of \citet[][Theorem~3.2.5]{van1996weak}.
Though we state the lemma in terms of maximizing $M_n$, 
an analogous conclusion applies for minimizing $M_n$.

\begin{lemma}
\label{thm:cr-rate}
Let $\{M_n(\theta): \theta \in \Theta \}$ be a stochastic process and $M(\theta)$ a deterministic function. 
Suppose $M(\theta) - M(\theta_0) \leq -\kappa d^2 (\theta, \theta_0)$ for some non-negative function $d: \Theta \times \Theta \to \R$ and positive number $\kappa$. 
Let $c_0$ be some value that may depend on $n$.
Suppose when $\eta \geq c_0$, we have
\[
\Pr \left\{ \sup_{\theta: d(\theta, \theta_0) \leq \delta} \left|(M_{\theta} - M)(\theta) - (M_n - M)(\theta_0)\right| \geq c_1 \delta^{\xi} \tau^{1/2} \right\} \leq e^{-\tau},
\]
where $\xi \in (0,1]$, $c_1$ is a constant which is independent of $\delta$ and $\tau$ but may depend on $n$. 

Let $\wh{\theta}_n = \argmax_{\theta \in \Theta} M_n(\theta)$. Define 
\[
\eta = \max \left\{4 \kappa^{-1/(2-\xi)} c_1^{1/(2-\xi)} \tau^{1/(4-2\xi)} , c_0 \right\}.
\]
Then, 
\[
\Pr \left\{ d(\wh{\theta}_n, \theta_0) \geq \eta \right\} \leq 3e^{-\tau}.
\]

\end{lemma}

\begin{proof}

Fix $\eta > 0$, define $\eta_j= \eta 2^{-j}$, $j \geq 0$, then
\[
\Pr \left\{ d(\wh{\theta}_n, \theta_0) \geq \eta \right\} \leq \sum^{\infty}_{j=1} \Pr \left[ \sup_{\theta: \eta_{j-1} \leq d(\theta, \theta_0) < \eta_j} \left\{ M_n(\theta) - M_n(\theta_0)\right\} \geq 0 \right].
\]
We observe that 
\begin{align*}
M_n(\theta) - M_n(\theta_0) &= \{ (M_n - M)(\theta) - (M_n - M)(\theta_0) \} + \{M(\theta) - M(\theta_0)\} \\
&\leq |(M_n - M)(\theta) - (M_n - M)(\theta_0)| - \kappa d^2(\theta, \theta_0).
\end{align*}
Hence, we have 
\begin{multline*} 
\Pr \left[ \sup_{\theta: \eta_{j-1} \leq d(\theta, \theta_0) < \eta_j} \{M_n(\theta) - M_n(\theta_0)\} \geq 0 \right] \\
\leq \Pr \left\{ \sup_{\theta: d(\theta, \theta_0) \leq \eta_j} |(M_n - M)(\theta) - (M_n - M)(\theta_0)| \geq \kappa \eta^2_{j-1} \right\}.   
\end{multline*}

Let $\beta = 1/(2-\xi)$. Then
$\eta = 4\kappa^{-\beta} c_1^{\beta}\tau^{\beta/2}$.  Hence,
$\eta^{2-\xi} \geq 4 \kappa^{-1} c_1 \tau^{1/2}$.  Because
$j 2^{-j} \leq 1$, $ j \geq j^{1/2} \geq 1$ for all $j\geq 1$ and
$\xi -2 \leq -1$,
\[
\eta^{2-\xi} \geq \kappa^{-1} 2^{-j+2}jc_1 \tau^{1/2} \leq \kappa^{-1} 2^{j(\xi-2) + 2}c_1 j^{1/2}\tau^{1/2}
\]
That is, 
$
\kappa \eta^2 2^{2j-2} \geq \eta^{\xi} 2^{j\xi} c_1 j^{1/2} \tau^{1/2}
$.
By the definition of $\eta_j$ and $\eta_{j-1}$, we have 
$
\kappa \eta^{2}_{j-1} \geq \eta^{\xi}_{j} c_1 j^{1/2} \tau^{1/2}.
$
By the condition on $M_n - M$, we have
\[
\Pr \left\{ \sup_{\theta: d(\theta, \theta_0) \leq \eta_j} |(M_n - M)(\theta) - (M_n - M)(\theta_0)| \geq \kappa \eta^{2}_{j-1} \right\} \leq e^{-j\tau}.
\]
Therefore, we have 
$
\Pr \left\{ d(\wh{\theta}_n, \theta_0) \geq \eta \right\}  \leq \sum_{j=1}^{\infty} e^{-j\tau} = e^{-\tau} / \left( 1- e^{-\tau} \right)
$.
Note that $e^{-\tau} / \left( 1- e^{-\tau} \right) \leq 3e^{-\tau}$ when $\tau \geq 1$ 
and $\Pr \left\{d(\wh{\theta}_n, \theta_0) \geq \eta \right\} \leq 1 \leq 3e^{-\tau}$ when $\tau < 1$. 

\end{proof}

\subsection{Proof of Theorem~1}

In this subsection, $\xi$ and $\beta$ denote arbitrary positive numbers.
The value of $\xi$ or $\beta$ may be different at each occurrence.
We start at the last stage $t=T$. 
Define $\varphi_T = r_T / (2 r_T + q_T)$.
Because $\wh{Y}_T = \wt{Y}_T$ for any $a \in \mathcal{A}_T$, 
under the conditions on $\bm{\gamma}_T$ and $\lambda_T$,
by Proposition~\ref{thm:krr-risk} and its corollary,
we have 
\[
\Pr \left[ \E_{\bm{X}} \left\{ \wh{Q}_T(\bm{X}, a) - Q_T(\bm{X}, a) \right\}^2 
\geq c_1 \left( n^{- 2\varphi_T + \xi} + n^{-1}\tau \right) \right] \leq e^{-\tau}.
\]
This establishes the consistency and convergence rate for $\wh{Q}_T$.

Next, we consider $(\wh{R}_{T\ell}, \wh{a}_{T\ell})$ for $\ell = 1, 2, \ldots$. 
In view of Assumption~4 (\romanNum{1}) and (\romanNum{2}), 
by reducing $\kappa$, we can have 
Assumption~4 (\romanNum{1}) hold for all $R$
instead of only those $R$ close to the true value.

When $\ell =1$, we have $\wh{G}_{T1} = G^{\ast}_{T1} = \mathcal{X}_T$. 
Thus, for any $a \in \mathcal{A}_T$, by equation~\eqref{eq:u-q} and
Lemma~\ref{thm:cr-global}, it follows that 
\[
\Pr \left\{ \sup_{R \in \mathcal{R}_T} |\Pn \wh{\Omega}_{T1} (R, a) - \E \Omega_{T1} (R, a)| \geq c_1 n^{-\varphi_T + \xi}\tau \right\} \leq e^{-\tau}.
\]
By Assumption~4~(\romanNum{3}), 
we have $\inf_{R \in \mathcal{R}_T, a \neq a^{\ast}_{T1}} \E \Omega_{T1} (R, a) \geq \E \Omega_{T1} (R^{\ast}_{T1}, a^{\ast}_{T1}) + \varsigma $. 
Thus, 
\begin{align*}
\Pr(\wh{a}_{T1} \neq a^{\ast}_{T1}) 
&\leq \sum_{a \neq a^{\ast}_{T1}} \Pr \left\{ \sup_{R \in \mathcal{R}_T} \Pn \wh{\Omega}_{T1}(R, a) \geq \Pn \wh{\Omega}_{T1} (R^{\ast}_{T1}, a^{\ast}_{T1})  \right\} \\
&\leq \sum_{a} \Pr \left\{ \sup_{R \in \mathcal{R}_T} \left| \Pn \wh{\Omega}_{T1} (R, a) - \E \Omega_{T1} (R, a) \right| \geq \varsigma/2  \right\}
\end{align*}
Hence, 
\[
\Pr (\wh{a}_{T1} \neq a^{\ast}_{T1}) \leq c_1 \exp(-c_2 n^{\varphi_T - \xi}), 
\]
where $c_1$ depends on $|\mathcal{A}_T|$ and $c_2$ depends on
$\varsigma$.  Actually, as seen from the proof of Theorem~2, we are
able to obtain a faster convergence rate for $\wh{a}_{T1}$.  However,
this does not affect the final result because $\wh{R}_{T1}$ converges
at a much slower rate, as shown below.

We proceed to establish the convergence rate for $\wh{R}_{T1}$. 
For any $\delta > 0$, by equation~\eqref{eq:u-q} and Lemma~\ref{thm:cr-local}, 
\begin{multline*}
\Pr \bigg\{ \sup_{R \in \mathcal{R}_T, \rho_T(R, R^{\ast}_{T1}) \leq \delta} 
 \left| \Pn \wh{\Omega}_{T1} (R, a^{\ast}_{T1}) - \Pn \wh{\Omega}_{T1} (R^{\ast}_{T1}, a^{\ast}_{T1}) 
 - \E \Omega_{T1}(R, a^{\ast}_{T1}) + \E \Omega_{T1} (R^{\ast}_{T1}, a^{\ast}_{T1}) \right| \\
\geq  c_1 \delta^{1/2 -\beta} n^{-\varphi_T + \xi} \tau \bigg\} \leq e^{-\tau}. 
\end{multline*}
Hence, by Lemma~\ref{thm:cr-rate},  
\[
\Pr \left\{ \rho_T(\wh{R}_{T1}, R^{\ast}_{T1}) \geq c_1 n^{-(2/3)\varphi_T + \xi } \tau \right\} \leq c_2 e^{-\tau}.
\]
Note that we take $\beta$ sufficiently small so that it can be absorbed into $\xi$.

We next proceed  to $\ell = 2$. 
By equation~\eqref{eq:u-q} and Lemma~\ref{thm:cr-global}, for any $a \in \mathcal{A}_T$, 
\[
\Pr \left\{ \sup_{R \in \mathcal{R}_T} 
|\Pn \wh{\Omega}_{T2} (R, a) - \E \Omega_{T2} (R, a) |  \geq c_1 n^{-(2/3) \varphi_T + \xi} \tau \right\} \leq e^{-\tau}. 
\]
Similar to $\wh{a}_{T1}$, we obtain 
\[
\Pr (\wh{a}_{T2} \neq a^{\ast}_{T2}) \leq c_1 \exp\left\{ -c_2 n^{(2/3)\varphi_T-\xi} \right\}.
\]
By equation~\eqref{eq:u-q} and Lemma~\ref{thm:cr-local}, for any $\delta > 0$, we have
\begin{multline*}
\Pr \bigg\{ \sup_{R \in \mathcal{R}_T, \rho_T(R, R^{\ast}_{T2}) \leq \delta} 
\left| \Pn \wh{\Omega}_{T2} (R, a^{\ast}_{T2}) - \Pn \wh{\Omega}_{T2} (R^{\ast}_{T2}, a^{\ast}_{T2}) 
- \E \Omega_{T2}(R, a^{\ast}_{T2}) + \E \Omega_{T2} (R^{\ast}_{T2}, a^{\ast}_{T2}) \right| \\
\geq  c_1 \delta^{1/2 -\beta} n^{-(2/3)\varphi_T + \xi} \tau \bigg\} \leq e^{-\tau}. 
\end{multline*}
Hence, by Lemma~\ref{thm:cr-rate}, 
\[
\Pr \left\{ \rho_T(\wh{R}_{T2}, R^{\ast}_{T2}) \geq c_1 n^{-(2/3)^2\varphi_T + \xi} \tau \right\} \leq c_2 e^{-\tau}.
\]
Again, $\beta$ is chosen to be sufficiently small so as to be absorbed into $\xi$.

Using induction, for any $\ell$, we obtain
\[
\Pr (\wh{a}_{T\ell}\neq a^{\ast}_{T\ell}) 
\leq c_1 \exp\left\{
-c_2 n^{(2/3)^{\ell-1} \varphi_T} \right\} 
\]
and 
\[
\Pr \left\{
\rho_T (\wh{R}_{T\ell}, R^{\ast}_{T\ell}) \geq c_1 n^{-(2/3)^\ell\varphi_T} \tau \right\} \leq c_2 e^{-\tau}.
\]
Make the change of variables $\tau \to c_1 n^{-(2/3)^\ell\varphi_T} \tau$, to obtain
\[
\Pr \{
\rho_T (\wh{R}_{T\ell}, R^{\ast}_{T\ell}) \geq \tau \} 
\leq c_1 \exp\left\{
-c_2 n^{(2/3)^{\ell} \varphi_T} \right\} .
\]

Therefore, 
\begin{align*}
\Pr \left\{ M_T(\wh{\pi}_T) \geq \tau \right\} 
&\leq \sum_{\ell=1}^{L_T^*} \pr\left(\wh{a}_{T\ell} \neq a^*_{T\ell}\right)
+ \sum_{\ell=1}^{L_T^*} \pr\left\{\rho_T(\wh{R}_{T\ell}, R^*_{T\ell}) \geq \tau / L_T^* \right\}\\
& \leq c_1 \exp(- c_2 n^{\phi_T - \xi} \tau),
\end{align*}
where $\phi_T = (2/3)^{L_T^*} \varphi_T$.
Consequently, 
\[
\Pr \left\{ V_T(\pi^{\ast}_{T}) - V_T(\wh{\pi}_T) \geq \tau \right\} 
\leq \Pr \left\{ M_T(\wh{\pi}_T) \geq \tau / B \right\} 
\leq c_3 \exp(- c_4 n^{\phi_T - \xi} \tau).
\]

We now proceed to the earlier stages. 
Consider the $(T-1)$th stage.
By the risk bounds of $\wh{Q}_{T}$ and $\wh{\pi}_{T}$,  
\[
\Pr \left\{ \Pn \left(\wh{Y}_{T} - \wt{Y}_{T} \right)^2 \geq c_1 n^{-\phi_T+\xi} \tau \right\} \leq c_2 e^{-\tau}.
\]
Hence, by Proposition~\ref{thm:krr-risk}, for any $a \in \mathcal{A}_{T-1}$, we have 
\[
\Pr\left[ E_{\bm{X}} \left\{ \wh{Q}_{T-1}(\bm{X}, a) - Q_{T-1}(\bm{X}, a)\right\}^2 \geq c_1 n^{-2\varphi_{T-1} + \xi} \tau \right] \leq c_2 e^{-\tau},
\]
where $\varphi_{T-1} = \min\{\phi_{T}/2, r_{T-1} /(2r_{T-1} + q_{T-1})
\}$, i.e., the convergence rate of $\wh{Q}_{T-1}$ depends on the kernel regression convergence rate assuming the true response $\wt{Y}$ is observed
and the convergence rate of the surrogate response $\wh{Y}$.

The analysis of $(\wh{R}_{T-1, \ell}, \wh{a}_{T-1, \ell})$s are the
same as in the last stage.
Thus, 
\[
\Pr \left\{ M_{T-1}(\wh{\pi}_{T-1}) \geq \tau \right\} 
\leq c_1 \exp(- c_2 n^{\phi_{T-1} - \xi} \tau),
\]
and
\[
\Pr \left\{ V_{T-1}(\pi^{\ast}_{T-1}) - V_{T-1}(\wh{\pi}_{T-1}) \geq \tau \right\}
\leq c_3 \exp(- c_4 n^{\phi_{T-1} - \xi} \tau),
\]
where $\phi_{T-1} = (2/3)^{L_{T-1}^*} \varphi_{T-1}$.
Using induction, these two inequalities hold when $T-1$ is replaced by $t = T-2, \dots, 1$.

\subsection{Proof of Theorem~2}

At the last stage, by Proposition~\ref{thm:krr-risk},
\[
\Pr \left[ \E_{\bm{X}} \left\{ \wh{Q}_T(\bm{X}, a) - Q_T(\bm{X}, a)\right\}^2 
\geq c_1 \left(n^{-2\varphi_T + \xi} + n^{-1} \tau\right) \right] 
\leq e^{-\tau},
\]
where $\varphi_T = r_{T}/(2r_{T} + q_{T})$ and $\xi>0$ is arbitrary. 
By equation~\eqref{eq:u-q}, 
\[
\Pr \left[ \E_{\bm{X}} \left\{ \wh{U}_T(\bm{X}, a) - U_T(\bm{X}, a)\right\}^2 
\geq c_1 \left(n^{-2\varphi_T + \xi} + n^{-1} \tau\right) \right] 
\leq e^{-\tau}.
\]

Using a similar argument to the proof of Theorem~1, 
\[
\Pr \left\{ \sup_{R \in \mathcal{R}_T} |\Pn \wh{\Omega}_{T1} (R, a) - \E \Omega_{T1} (R, a)| \geq c_1 \left(n^{-\varphi_T + \xi} + n^{-1/2} \tau^{1/2} \right) \right\} \leq e^{-\tau}.
\]
and
\[
\Pr(\wh{a}_{T1} \neq a^{\ast}_{T1}) 
\leq \sum_{a} \Pr \left\{ \sup_{R \in \mathcal{R}_T} \left| \Pn \wh{\Omega}_{T1} (R, a) - \E \Omega_{T1} (R, a) \right| \geq \varsigma/2  \right\}.
\]
Note that $\varsigma$ is a fixed number independent of $n$.
Let $\tau^{1/2} = n^{1/2} \max(c_2 \varsigma - n^{-\varphi_T + \xi}, 0)$
and choose $c_2$ such that $2 c_1 c_2 < 1$.
Then, 
\[
\Pr (\wh{a}_{T1} \neq a^{\ast}_{T1} ) \leq c_3 \exp(-c_4 n)
\]
as $\varphi_T \in (0, 1)$.

Define
$\vartheta = \inf_{R: \rho_{T}(R, R^{\ast}_{T1}) > 0} \rho_T(R,
R^{\ast}_{T1})$.
Because the covariates are discrete, $\vartheta$ is strictly positive.
This is a major difference between the continuous covariates and the
discrete covariates.  By Assumption~4 (\romanNum{1}), we have
\begin{align*}
\Pr\left\{ \rho_T(\wh{R}_{T1} , R^{\ast}_{T1}) > 0 \right\}
&\leq \Pr \left\{ \sup_{R \in \mathcal{R}_T: \rho_T(R, R^{\ast}_{T1}) \geq \vartheta} \Pn \wh{\Omega}_{T1}(R, a^{\ast}_{T1})\geq \Pn \wh{\Omega}_{T1}(R^{\ast}_{T1}, a^{\ast}_{T1}) \right\} \\
&\leq \Pr \left\{ \sup_{R\in \mathcal{R}_T} \left| \Pn \wh{\Omega}_{T1}(R, a^{\ast}_{T1}) - \E \Omega_{T1} (R, a^{\ast}_{T1}) \right| \leq \kappa \vartheta^2 /2 \right\} \\
&\leq c_5 \exp(-c_6 n).
\end{align*}

We next analyze $(\wh{R}_{T2}, \wh{a}_{T2})$.
For any $a \in \mathcal{A}_T$, 
\[
\Pr \left\{ \sup_{R \in \mathcal{R}_T} |\Pn \wh{\Omega}_{T2} (R, a) - \E \Omega_{T2} (R, a)| \geq c_1 \left(n^{-\varphi_T + \xi} + n^{-1/2} \tau^{1/2} \right) \right\} \leq e^{-\tau}.
\]
Similar to $(\wh{R}_{T1}, \wh{a}_{T1})$, 
\[
\Pr (\wh{a}_{T2} \neq a^{\ast}_{T2} ) \leq c_1 \exp(-c_2 n)
\]
and
\[
\Pr\left\{ \rho_T(\wh{R}_{T2} , R^{\ast}_{T2}) > 0 \right\} \leq c_3 \exp(-c_4 n).
\]

As seen from this inequality,
a notable difference is that the estimation error does not propagate along the list,
compared to the general case where covariates can be continuous. 
The tail probability decays at the same exponential rate for every~$\ell$.
Therefore, we have
\[
\Pr \left\{ M_T(\wh{\pi}_T) > 0 \right\} 
\leq \sum_{\ell=1}^{L_T^*} \pr\left(\wh{a}_{T\ell} \neq a^*_{T\ell}\right)
+ \sum_{\ell=1}^{L_T^*} \pr\left\{\rho_T(\wh{R}_{T\ell}, R^*_{T\ell}) > 0 \right\}
\leq c_1 \exp(-c_2 n).
\]
Thus, 
\[
\Pr\left\{ V_T(\pi^{\ast}_{T}) - V_T(\wh{\pi}_T) > 0\right\}
\leq \Pr \left\{ M(\wh{\pi}_T) > 0 \right\}
\leq c_1 \exp(-c_2 n).
\]

We then move to the $(T-1)$th stage. 
Conditional on the event $\{M(\wh{\pi}_T) = 0\}$, which occurs with
probability $1- c_1\exp(-c_2 n)$,  
\[
\Pr \left[ \wh{Q}\left\{ \bm{X}_T, \wh{\pi}_T(\bm{X}_T) \right\} = \wh{Q} \left\{ \bm{X}_T, \pi^{\ast}_T(\bm{X}_T)\right\} \right] = 1.
\]
Hence,
\[
\Pr \left\{ \Pn \left(\wh{Y}_{T-1} - \wt{Y}_{T-1}\right)^2 
\geq c_1 \left(n^{-2\varphi_T + \xi} + n^{-1} \tau \right) \right\} 
\leq e^{-\tau}. 
\]

Define $\varphi_{T-1} = \min\left\{ r_{T-1}/(2r_{T-1} + q_{T-1}), \varphi_T \right\}$. By Proposition~\ref{thm:krr-risk},  
\[
\Pr \left[ \E_{\bm{X}} \left\{ \wh{Q}_{T-1}(\bm{X}, a) - Q_{T-1}(\bm{X}, a)\right\}^2 
\geq c_1 \left(n^{-2\varphi_{T-1} + \xi} + n^{-1} \tau\right) \right] 
\leq e^{-\tau}.
\]
Note that nothing is changed except that $T$ is replaced by $T-1$.
Using the same approach as in the $T$th stage, 
conditional on the event $\{M(\wh{\pi}_T) = 0\}$,
we obtain
\[
\Pr \left\{ M_{T-1}(\wh{\pi}_{T-1}) > 0 \right\} 
\leq c_1 \exp(-c_2 n),
\]
and
\[
\Pr\left\{ V_{T-1}(\pi^{\ast}_{T-1}) - V_{T-1}(\wh{\pi}_{T-1}) > 0\right\}
\leq c_1 \exp(-c_2 n).
\]
Because the event $\{M(\wh{\pi}_T) = 0\}$ occurs with probability $1- c_1\exp(-c_2 n)$,
both inequalities hold unconditionally with larger constants $c_1$ and $c_2$.

Using induction, we can establish analogous inequalities for $t = T-2, \dots, 1$.

\section{Algorithm Details and Proof of Proposition~1}

Fix an $t$ and $\ell$.
Define 
\[
U_{i a t \ell} = \left[ \wh{Q}_t \left\{ \bm{X}_{it}, \wh{\pi}_t^{Q}(\bm{X} _{it})\right\} - \wh{Q}_t(\bm{X}_{it}, a) - \zeta \right]
I \left( \bm{X}_{it} \in \wh{G}_{t\ell} \right).
\]
For notation simplicity, we shall omit the subscript $t$ and $\ell$
and write $U_{i a}$ and $\bm{X}_i$.
By definition of $(\wh{R}_{t\ell}, \wh{a}_{t\ell})$, 
\[
(\wh{R}_{t\ell}, \wh{a}_{t\ell}) = \argmin_{R \in \mathcal{R}_t, a \in \mathcal{A}_t} \frac{1}{n} \sum_{i=1}^n U_{ia}I(\bm{X}_{i} \in R) - \eta\{2-V(R) \}. 
\]
We will first fix the treatment $a$ and the covariates involved in $R$,
and focus on the computation of the optimal thresholds.
Then we will loop over all covariate pairs and all treatment options.

\paragraph{Finding the threshold when $R$ involves one covariate}

Without loss of generality, we assume $R = \{\bm{x}: x_j \leq \tau\}$.
The other situation $R = \{\bm{x}: x_j > \tau\}$ can be handled similarly. 
We want to compute
\[
\wh{\tau} = \argmin_{\tau} \sum_{i=1}^n U_{ita} I(X_{ij} \leq \tau),
\]
where $X_{ij}$ is the $j$th component of $\bm{X}_{i}$. 

Let $i_1, \dots, i_n$ be a permutation of $1, \dots, n$ such that $X_{i_1 j}\leq \dots \leq X_{i_n j}$. 
Because the objective function is piecewise constant,
we only need to compute 
\[
F(\tau) = \sum_{i=1}^n U_{ia} I(X_{ij} \leq \tau)
\]
when $\tau$ equals to some $X_{i_s j}$. 
We observe that
\[
F(X_{i_s j}) = \sum_{h \leq s} U_{i_h a}.
\]
Thus, it is clear that when $s \geq 2$
\[
F(X_{i_s j}) = F(X_{i_{s-1} j}) + U_{{i_s}a}.
\]

Hence, by starting at $s=1$ and using the recursive relationship, we can compute $F(X_{i_{s} j})$ for all $s$ and pick the smallest one in $O(n)$ time.

\paragraph{Dealing with ties}

If $ X_{i_s j} = X_{i_{s+1} j} $ for some $s \geq 1$,
then $F(X_{i_{s} j})$ should not be counted
when picking the minimum.
This is because $F(X_{i_{s} j})$ has not  included all subjects with
$X_{ij} = X_{i_s j}$ yet.

To avoid this problem, when there are ties, 
we first aggregate the $U_{ia}$ values
for subjects having the same value of $X_{ij}$.
Similar action can be taken when $R$ involves two covariates,
in which case the $U_{ia}$ values for subjects having the same value 
for both covariates are aggregated.

\paragraph{Finding the threshold when $R$ involves two covariates}
This situation is more complicated.
Without loss of generality, we assume $R=\{\bm{x}: x_j \leq \tau \text{ and } x_k \leq \sigma\}$. 
We want to compute 
\[
(\wh{\tau}, \wh{\sigma}) = \argmin_{\tau, \sigma} \frac{1}{n} \sum_{i=1}^n U_{ia} I(X_{ij} \leq \tau, X_{ik} \leq \sigma). 
\]
We cannot utilize the idea for one covariate
as there is no natural ordering in two-dimensional space. 
Our solution is to sort in one dimension and 
to use binary tree for fast lookup and insertion in the other dimension. 

We start with constructing a complete binary tree of at least $n$ leaves.
The height of such a tree is of order $O(\log_2 n)$. 

Let $i_1, \dots, i_n$ be a permutation of $1, \dots, n$ 
such that $X_{i_1 j} \leq \dots X_{i_n j}$. 
At each time $s$, we will insert $U_{i_s a}$ into the binary tree and search for the optimal threshold $\sigma$
among $X_{ik}, i = 1, \dots, n$. 
Note that at time $s$, values $U_{i_h a}, h \leq s$ are contained in
the binary tree.
So we are looking at the threshold $\tau = X_{i_s j}$.
Specifically, if the rank of $X_{i_s k}$ among $X_{ik}$s is $h$, 
which means $X_{i_s k}$ is the $h$th smallest among $X_{ik}$s, 
then we put $U_{i_s a}$ in the $h$th leaf from the left in the tree. 

In the tree, each node is associated with a subtree in which
that node serves as the root.
Each node contains two pieces of information. 
First, it computes the sum of all $U_{i_s a}$s in the associated subtree. 
Second, it computes the best thresholding sum in the associated subtree,
which is the smallest value among the sum of all $U_{i_s a}$s that satisfies $X_{i_s k} \leq \sigma$ for some $\sigma$,
where $\sigma$ can take the value of any $X_{i_s k}$ in the associated subtree. 

The binary tree structure enables us to update these two pieces of information effectively when a new value, $U_{i_s a}$, is inserted into the tree. 
We move from the leaf node to its parent, and then its ancestors, and finally the root. 
At each node, the sum of all $U_{i_s a}$s in the associated subtree is increased by $U_{i_s a}$. 
As for updating the best thresholding sum, because the thresholding condition is $X_{i_s k} \leq \sigma$, 
the best thresholding sum of a node can only be either 
the best thresholding sum in its left child,
or,
the sum of all $U_{i_s a}$ values in the left child plus the best thresholding sum in the right child, 
whichever is smaller. 

Because the height of the tree is $O(\log_2 n)$, the updating process
involves at most $O(\log_2 n)$ nodes and the time complexity at each
node is constant.  Therefore, when $U_{i_s a}$ is inserted into the
tree, we are able to find the optimal $\sigma$ that minimizes
$\sum_{h \leq s} U_{i_h a} I(X_{i_h k} \leq \sigma)$ in $O(\log_2 n)$
time.

Then we let $s$ run from $1$ to $n$, and find the $s$ that gives the minimum. 
In this way, we find the minimum of $\sum_{h=1}^{n} U_{i_h a} I(X_{i_h k} \leq \sigma, X_{i_h j} \leq X_{i_s j})$
with respect to $\sigma$ and $s$, 
which is exactly the minimum of $\sum_{i=1}^n U_{i_s a}I(X_{ik} \leq \sigma, X_{ij} \leq \tau)$ with respect to $\sigma$ and $\tau$. 
And the time complexity for finding both $\tau$ and $\sigma$ is $O(n \log_2 n)$.

\paragraph{Finding the covariate(s) and treatment}

Heretofore, we have discussed how to find the optimal thresholds when
the covariates to use $X_{ij}$, $X_{ik}$ and the treatment $a$ are
given.  Certainly we need to explore all $R$s defined using only one
covariate, and all $R$s defined using some pair of $X_{ij}$ and
$X_{ik}$.  We also need to loop over all treatment options
$a \in \mathcal{A}_t$.

Therefore, the overall time complexity is $O(n \log n q^2_t m_t)$, where $q_t$ is the dimension of $\bm{X}_i$ and $m_t = |\mathcal{A}_t|$ is the number of available treatment options.

\section{Variables in Data Analysis}

In the first stage, we have the following variables:
\begin{enumerate}
\item age: integer;
\item gender: 1 for male, 0 for female;
\item race: 1 for white, 0 for others;
\item education level: 1 for high school or below, 2 for some college, 3 for bachelor or up;
\item work status: 1 for full time, 0.5 for part time, 0 for no work;
\item bipolar type: 1 or 2;
\item status prior to the onset of the current episode: 1 for remission longer than 8 weeks;
\item status prior to the onset of the current episode: 1 for manic/hypomanic;
\item status prior to the onset of the current episode: 1 for mixed/cycling;
\item SUM-D at week 0;
\item SUM-ME at week 0.
\end{enumerate}

In the second stage, we have the following variables:
\begin{enumerate}
\item binary indicator for adverse effect tremor;
\item binary indicator for adverse effect dry mouth;
\item binary indicator for adverse effect sedation;
\item binary indicator for adverse effect constipation;
\item binary indicator for adverse effect diarrhea;
\item binary indicator for adverse effect headache;
\item binary indicator for adverse effect poor memory;
\item binary indicator for adverse effect sexual dysfunction;
\item binary indicator for adverse effect increase appetite;
\item SUM-D at week 6;
\item SUM-ME at week 6.
\end{enumerate}

\end{document}